\documentclass[12pt, draftclsnofoot, onecolumn]{IEEEtran}

\usepackage{amsfonts}
\usepackage{amssymb}
\usepackage{stfloats}
\usepackage{cite}
\usepackage{graphicx}
\usepackage{psfrag}
\usepackage{subfigure}
\usepackage{amsmath}
\usepackage{array}
\usepackage{epstopdf}
\usepackage{authblk}
\usepackage{graphicx} 
\usepackage{amsthm} 
\usepackage{lipsum}
\usepackage{verbatim} 
\usepackage{authblk}
\usepackage{algorithm}
\usepackage{algpseudocode}
\usepackage{amsmath}
\usepackage{graphics}
\usepackage{epsfig}
\usepackage{setspace}

\usepackage[usenames]{color}

\newtheorem{Definition}{Definition}
\newtheorem{Problem}{Problem}
\newtheorem{Lemma}{Lemma}

\newtheorem{Algorithm}{Algorithm}

\newtheorem{Baseline}{Baseline}
\newtheorem{Example}{Example}
\newtheorem{Model}{Model}

\title{Joint Content Delivery and Caching Placement via Dynamic Programming}

\author{Bojie Lv, Lexiang Huang and Rui Wang\thanks{Bojie Lv, Lexiang Huang and Rui Wang are with Department of Electrical and Electronic Engineering, The Southern University of Scienece and Technology, China, Email: \{lvbj@mail.sustc.edu.cn, huanglx@mail.sustc.edu.cn, wang.r@sustc.edu.cn\}
		
Part of this work has been submitted to IEEE ICC 2018 \cite{WANG2017}. We have extended the conference paper by including the learning-based algorithm in Section IV-C and proactive scheduling algorithm in Section V, and adding more illustrative simulation results.}}

\begin{document}

\maketitle
\begin{abstract}
In this paper, downlink delivery of popular content is optimized with the assistance of wireless cache nodes. Specifically, the requests of one file is modeled as a Poisson point process with finite lifetime, and two downlink transmission modes are considered: (1) the base station multicasts file segments to the requesting users and selected cache nodes; (2) the base station proactively multicasts file segments to the selected cache nodes without requests from users. Hence the cache nodes with decoded files can help to offload the traffic upon the next file request via other air interfaces, e.g. WiFi. Without proactive caching placement, we formulate the downlink traffic offloading as a Markov decision process with random number of stages, and propose a revised Bellman's equation to obtain the optimal control policy. In order to address the prohibitively huge state space, we also introduce a low-complexity sub-optimal solution based on linear approximation of the value functions, where the gap between the approximated value functions and the real ones is bounded analytically. The approximated value functions can be calculated from analytical expressions given the spatial distribution of requesting users. Moreover, we propose a learning-based algorithm to evaluate the approximated value functions for unknown distribution of requesting users. Finally, a proactive caching placement algorithm is introduced to exploit the temporal diversity of shadowing effect. It is shown by simulation that the proposed low-complexity algorithm based on approximated value functions can significantly reduce the resource consumption at the base station, and the proactive caching placement can further improve the performance.
\end{abstract}

\section{introduction}\label{sec:intro}
Caching is a promising technology to improve the network spectral efficiency \cite{Leung2014} or cut down energy consumption \cite{ISWCS,Vincent2013} by predicting the potential content requests of users. In this paper, we consider a flexible deployment scenario where there is no wired connection or dedicated spectrum between the base station (BS) and cache nodes. The cache nodes have to update their buffers via downlink, sharing the same transmission resources with ordinary users. Hence the joint content delivery and caching placement problem shall be addressed to minimize the downlink resource consumption.

\subsection{Related Works}

With the wired connection between cache nodes and BSs, the buffer capacity limitation at the cache nodes becomes the major performance bottleneck, and there have been a number of research efforts spent on the file placement of the cache nodes. For example in \cite{Debbah2014}, it was shown that prediction of future demand will improve the performance of file placement at the cache nodes. The authors in\cite{B.xia} evaluated the performance of wireless heterogeneous networks, where the small BSs (cache nodes) cache the most popular files under the buffer capacity limitation at the cache nodes. Alternatively, the authors in\cite{W.Choi2016} showed that caching the files randomly with optimized probabilities is better than saving the most popular files when each user can be served by multiple cache nodes. In order to offload more traffics to cache nodes, a file placement algorithm was derived in \cite{Leung2016} with the constrains of cache nodes' capacity and backhaul rates between the cache nodes and macro BS. In cooperative interference networks, the authors proposed a short-term MIMO precoding design and a long-term file placement policy to minimize the transmission power \cite{Vincent2013}. In\cite{coded_cache_1,M.Tao2017}, the authors proposed multicast algorithms with coded caching scheme. All the above works assumed that there are wired links between the cache nodes and BSs. However, it might be costly to deploy cache nodes with wired connection in some areas, and hence the caching placement via wireless links (e.g., cellular downlink) should also be studied \cite{Survey2015}. In \cite{WCNC2017}, the transmissions from BS to cache nodes and from cache nodes to users shared the same spectrum. However, it neglected the possibility that both cache nodes and users could listen to the BS simultaneously via a multicast mode. In fact, if file caching has to be made via downlink, the phase of caching placement can be coupled with the phase of content delivery. For example, at the first few transmissions of one file, both cache nodes and requesting users could listen to the BS simultaneously; and the cache nodes can help to forward the file as long as they have been able to decode it. Thus the downlink strategy should be optimized spanning the whole lifetime of a file (including the transmission to both cache nodes and users). However, this has not been addressed by the existing literature.

When the caching placement shares the same transmission resource as the downlink of ordinary users, there will be a multi-stage optimization issue. For example, if more transmission resource is spent on caching placement in the previous stages, more traffic may be offloaded in the following stages, and vice versa. Thus the transmission dynamics in time domain should be captured to minimize the overall downlink resource consumption, and the method of dynamic programming can be utilized. In fact, the dynamic programming via Markov decision process (MDP) has been considered in the delay-aware resource allocation of wireless systems. For example, the infinite-horizon MDP has been used to optimize the cellular uplink \cite{Moghadari2013,cui2010} and downlink transmissions \cite{cui2017_01}, and relay networks\cite{Wang2013}, where the average transmission delay is either minimized or constrained. Moreover, the low-complexity algorithm design is usually considered in the above works to avoid the curse of dimensionality \cite{Shewhart2011Approximate}. However, the popular files to be buffered at the cache nodes should usually have a finite lifetime, and the infinite-horizon MDP may not be suitable in modeling anymore. Nevertheless, the MDP with finite stages is usually more complicated \cite{FHMDP}, and it is still an open issue on the low-complexity algorithm design with finite-horizon MDP. 

\subsection{Our Contributions}

In this paper, we consider the downlink file transmission with the assistance of cache nodes, where the caching placement shares the downlink resource with ordinary users. Specifically, two downlink transmission modes are considered: (1) the BS multicasts file segments to the requesting users and selected cache nodes; (2) the BS proactively multicasts file segments to the selected cache nodes without requests from users (namely proactive caching placement). The communication links between cache nodes and users are via different spectrum from the downlink (e.g., Wi-Fi) as \cite{May2016,Molisch2016}, and the BS tries to minimize the average downlink resource consumption by offloading traffic to cache nodes. The main contributions of this work are summarized below:
\begin{itemize}
	\item Without proactive caching placement, we model the downlink transmission of one file within his lifetime as a Markov decision process (MDP) with random number of stages. Note that this is not conventional MDP problem with finite and fixed number of stages, we propose a revised version of Bellman's equation, where the optimal control policy can be obtained given the value functions. Then we introduce a linear approximation on the value functions so that the exponential complexity can be reduced into linear. With the knowledge of spatial distribution of requesting users, the approximated value functions can be calculated via analytically expression; whereas, a learning algorithm is also introduced to evaluate the approximated value functions with unknown distribution of requesting users.
	
	\item The bounds on the approximation error of value functions are also obtained. Based on it, we further derive the bounds on the overall average transmission cost of the BS.
	
	\item A per-stage optimization approach is proposed to include the proactive caching placement into the above optimization framework.
\end{itemize}
It is shown by simulation that the proposed low-complexity algorithm based on approximated value functions can significantly reduce the resource consumption at the base station, and the proactive caching placement can further improve the performance.

The remainder of this paper is organized as follows. In Section \ref{sec:model}, the system model is introduced. In Section \ref{sec:opt}, we formulate the downlink resource allocation without proactive caching placement as a MDP with random number of stages, and introduce the optimal solution. In Section \ref{sec:approximation}, a linear approximation is proposed to the value functions, and the bound on approximation error is derived. In Section \ref{sec:proactive}, a per-stage optimization approach is proposed to include the proactive caching placement. The numerical simulation is provided in Section \ref{sec:sim}, and the conclusion is drawn in Section \ref{sec:con}.

\section{System Model}\label{sec:model}

\begin{figure}[tb]
	\centering
	\includegraphics[scale = 0.5]{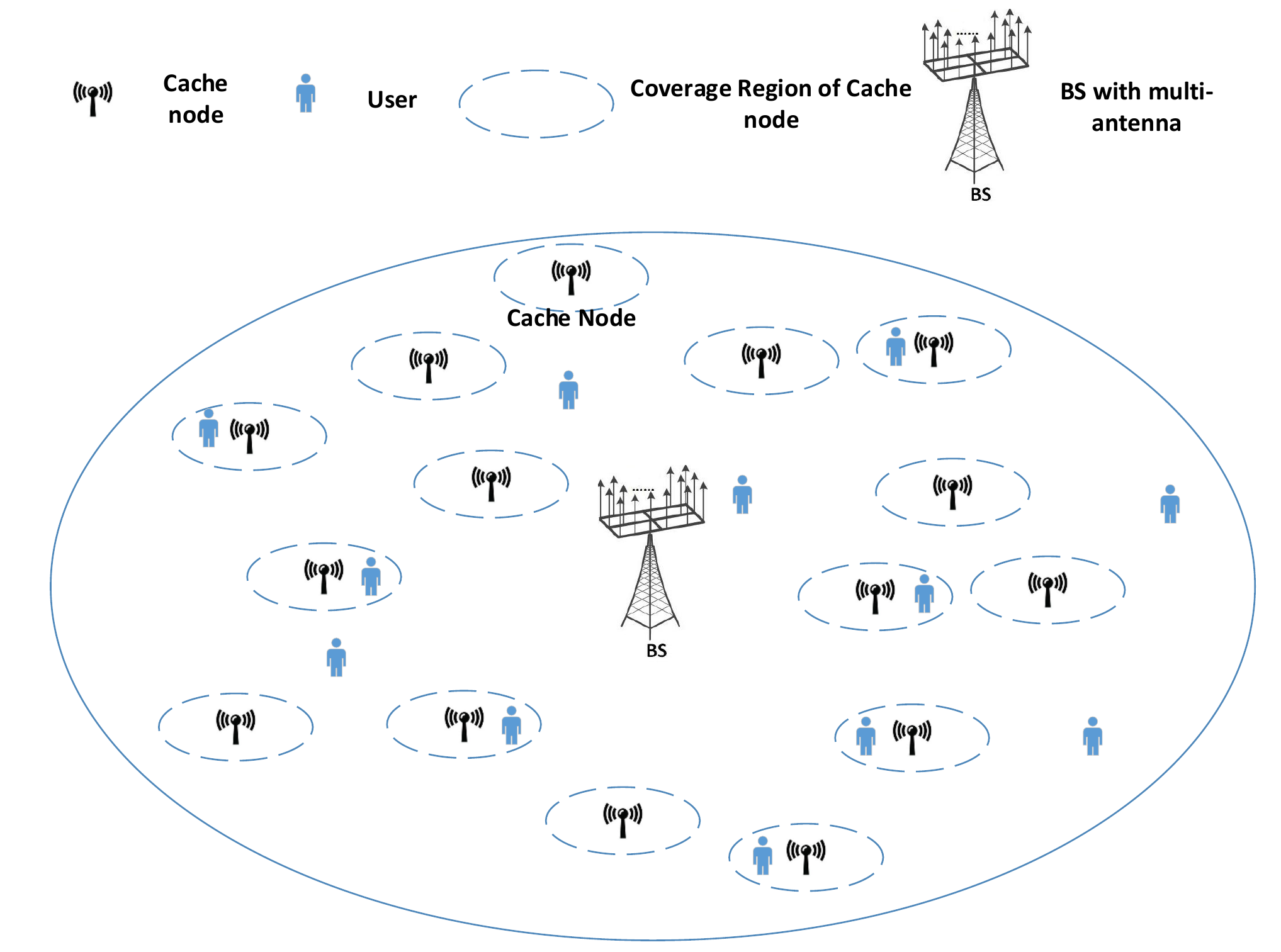}
	\caption{Illustration of network model with one BS and multiple wireless cache nodes.}
	\label{fig:scheme}
\end{figure}

\subsection{File Request Model}

As illustrated in Fig. \ref{fig:scheme}, the downlink transmission in a cell with one multi-antenna BS, $N_C $ single-antenna cache nodes is considered. Let $ N_T $ be the number of antennas at the BS. Without loss of generality, it is assumed that the BS locates in the origin, and the locations of cache nodes are denoted as $ \mathbf{c}_1, \mathbf{c}_2, ..., $ and $ \mathbf{c}_{N_C} $ respectively, which can be arbitrary in the cell coverage. Let $ \mathcal{C}_i $ ($ \forall i=1,2,...,N_C $) be the coverage region of the $ i $-th cache node, and $ \mathcal{C} = \mathcal{C}_i \cup \mathcal{C}_i \cup ... \cup \mathcal{C}_{N_C} $.

It is reported in the existing literature that a significant portion of the downlink data may be requested by multiple users of one cell, e.g. popular videos, web pages and etc, which can be buffered at the cache nodes for future transmissions. This type of data is usually organized as files. We shall focus on the offloading of downlink resource for the these files. However, we do not exclude the possibility that BS may simultaneously transmit data dedicated to one certain user, which is not necessarily buffered at the cache nodes. Without loss of generality, it is assumed that each file consists of $ R_F $ information bits, which are further divided into $ N_S $ segment equally. 

In most of the existing literature on wireless caching networks, the popularity of downlink files is modeled by an access probability, which may follow the Zipf distribution \cite{zipf_1}. This model is applicable when the caching placement and content delivery to users are scheduled with separated resource (i.e. different time or spectrum). However, these two phases are coupled in this paper. Specifically, the scenario we consider is described below:
\begin{itemize}
	\item There is no wired or dedicated wireless connection between cache nodes and BS. Hence the cache nodes can only receive files via downlink. 
	
	\item The downlink files can be requested by users since it is generated, and the downlink transmission to users and cache nodes can happen simultaneously. Hence there is no dedicated phase of caching placement as assumed in the existing literature.
\end{itemize}
For instance, users may not find the desired file from the nearby cache nodes in the very beginning of the file's lifetime (the cache nodes have not decoded the file). After certain times of downlink file transmission, the situation can be different. This is because some cache nodes may have successfully decoded the above file. Thus the following file transmission is strongly related to the previous transmissions. In order to capture this temporal dynamics of the coupled user requesting and caching placement, we introduce the following file request model based on Poisson point process (PPP).

\begin{Model}[File Request Model]
	
For the elaboration convenience, it is assumed that each file has the same lifetime $ T $. Suppose the $ f $-th file ($ f=1,2,... $) is available for access since time instance $ t_f $, we consider the requests on this file happening during the time period $ [t_f, t_f+T] $. The Poisson point process is adopted to model the event of file request within a file's life time, and $ \lambda_f $ is the process intensity of the $ f $-th file. Hence the probability mass function of request number on the $ f $-th file ($ \forall f $), when the remaining lifetime is $ T_{rem} $, is given by
\begin{equation} \label{eqn:request}
\Pr(\mbox{Request Number} = n) = \frac{(\lambda_f T_{rem})^n}{n!} e^{-\lambda_f T_{rem}}.
\end{equation}

The locations of new requesting users are independently and identically (i.i.d.) distributed following certain spatial distribution in the cell coverage. It is assumed that users' locations do not change during the file transmission, and the users become inactive after receiving their desired files.
\end{Model}

Note that Poisson point process has been widely used to model the phone call at an exchange. Moreover, given a time duration, the above model could provide an probability that one file is requested by at least one user, which degenerates to the access model in the existing literature. For example, the probability that the $ f $-th file is requested by at least one user within a time duration $ \tau  $ is given by 
\begin{equation*}
P_f = 1 - \sum_{n=1}^{\infty} \frac{(\lambda_f \tau)^n}{n!} e^{-\lambda_f \tau}.
\end{equation*}

In this paper, we shall first assume that the spatial distribution of new requesting users is known to the BS, and derive a low-complexity resource scheduling algorithm. Then we shall continue to consider the scenario that the BS does not know the distribution of new requesting users, and introduce a learning-based algorithm to adapt the proposed low-complexity resource scheduling algorithm.

\subsection{File Transmission with Wireless Caching}

Since the cache nodes rely on downlink to obtain the popular files and each file can be requested at the very beginning of its lifetime, the phases of caching placement and user requesting are coupled, and it is necessary to consider them jointly (Note that these two phases are usually considered in a separated way in most of the existing literature \cite{AAAA,tao2017}). Specifically, there are three types of communication links in the cellular networks with wireless cache nodes, which are introduced below.
\begin{itemize}
	\item {\bf Requested downlink multicast:} the BS delivers the requested file segments to user and the chosen cache nodes simultaneously. This will happen when the requested segments cannot be found in the nearby cache nodes. For example, if one user is within the coverage of $ \mathcal{C}_i $, it will receive the file segments from the BS when these file segments cannot be found in the $ i $-th cache node. Since both requesting user and cache nodes can listen when the BS is transmitting, the transmission mode is multicast.
	
	\item  {\bf Device-to-device (D2D) communications:} the cache nodes forward the requested file segments to users. This will happen when the requested segments hit the buffer of nearby cache nodes. For example, if one user is within the coverage of $ \mathcal{C}_i $, it will receive the file segments from the $ i $-th cache node when the file segments can be found there. This communication is made directly from cache nodes to users. It can use Wi-Fi, bluetooth, or other air interfaces, which is not in the same spectrum as downlink. For example, the offloading from cellular to Wi-Fi has attracted a number of attenuations \cite{May2016} \cite{Molisch2016}.
	
	\item {\bf Proactive caching placement:} the BS delivers some file segments to cache nodes without the request from users. It is assumed that the BS will determine whether to deliver one file segment, which is still in its lifetime, to cache nodes proactively in every $ T_p $ seconds.
\end{itemize}

Without proactive caching placement, the first transmission of each file is made directly from the BS, and some cache nodes may also be able to decode the whole file or some segments. For the successive requests of the same file, user will receive one segment from the cache node if the following two conditions are satisfied: (1) the user is in the coverage of certain cache node; (2) the aforementioned file segment has been successfully decoded by the aforementioned cache node. If more transmission resource (transmission power or symbols) is spent in the previous rounds of file transmission, more cache nodes would be able to buffer the file, which may save the transmission resource of the BS in the future. Clearly, it is necessary to exploit the correlation among different transmissions of the same file in downlink resource allocation. Furthermore, if proactive caching placement is allowed, the BS may multicast certain file segments to the chosen cache nodes according to their downlink channels even without the requests from users. This decision might be made if the BS believes the proactive caching placement can save more transmission resource in the future transmission.

\begin{Example}[Downlink File Delivery via Wireless Caching]
Consider the following example: the downlink file is divided into two segments; there are two users A and B, which will raise the requests on this file in sequence; and User B is in the coverage of cache node (say Cache Node C). Suppose when User A requests the file, the Cache Node C has not buffered it. The BS chooses User A and Cache Node C as the downlink receivers. Hence when the User B requests the same file, it can be served by Cache Node C.
	
Moreover, if the proactive caching placement is allowed, the BS may find the channel to the Cache Node C is very good, and deliver the first segment of the file to the Cache Node C before the request from User A. When the User A requests the first file, it can obtain the first segment from the Cache Node C and another segment from the BS. 
\end{Example}

Since the data rate of wireless communications is usually much smaller than wired, and the received files can be removed from the buffer of cache nodes when it is expired, we ignore the limitation on capacity  of cache buffer in this paper \footnote{For example, suppose that one BS is transmitting downlink files with overall data rate of $ 1 $ Gbps, and the lifetime of each file is $ 24 $ hours. Then the maximum required storage capacity of one cache node is around $ 10 $T bytes, which is a mild requirement for ignoring the buffer capacity limitation. }. Moreover, since the D2D communications can be done distributively and parallelly with relatively low transmission power, we shall focus on the transmission resource consumption at the BS only, which is the bottleneck of the overall system.

\subsection{Downlink Physical Layer Model}

In downlink, the receivers include the requesting user and cache nodes, and the space-time block code (STBC) with full diversity is used at the BS to facilitate the multicast communications. There are two main reasons that STBC is chosen for downlink multicast: (1) the BS need not to collect channel state information (CSI); (2) the full diversity can be achieved at all the receivers. 

It has been well studied in the existing literature that the user selection strategy in downlink frames (the typical duration of one frame is $ 10 $ milliseconds) will affect the transmission delay. For example, the delay-aware downlink control has been investigated in \cite{cui2012},\cite{Chamola2017}. In this paper, however, we shall consider the resource allocation in the time scale of file segment transmission, which consists of thousands of frames. In order to avoid the prohibitively complicated optimization on multiple time scales, we make the following simple and practical assumption: the transmission of one file segment is within the coherent time of shadowing attenuation, and the shadowing attenuation for the transmission of different segments of one file is i.i.d. Based on this assumption, we shall consider the optimization of cache node selection, transmission power level and transmission time of each file segment for both requested downlink multicast and proactive caching placement. Note that when the above three parameters are determined for each segment transmission, the user selection for each frame can be further considered in the finer time scale to optimize the average transmission delay, which is outside the scope of this paper.

In the requested downlink multicast, we shall refer to the user, which raises the $ n $-th request on the $ f $-th file, as the $ (f,n) $-th user, and refer to the $ s $-th segment of the $ f $-th file as the $ (f,s) $-th segment. Since the transmission time of one file segment is much larger than the channel coherent time of small-scale fading, it is assumed that the ergodic channel capacity span all possible small-scale channel fading can be achieved during one segment transmission. Let $ \rho_{f,n} $ and $ \rho_c $ be the pathloss from the BS to the $ (f,n) $-th user and the $ c $-th cache node respectively, $ \eta_{f,n,s} $ and $ \eta_{f,n,s}^c $ be the corresponding shadowing attenuation in $ n $-th transmission of the $ (f,s) $-th segment. In the downlink transmission of one file, the location of requesting user is assumed to be static. Let $ P_{f,n,s} $ be the downlink transmission power of the $ s $-th file segment in response to request of the $ (f,n) $-th user, $ N_{f,n,s} $ be the number of downlink transmission symbols scheduled to deliver the $ s $-th segment to the $ (f,n) $-th user, the throughput achieved by the $ (f,n) $-th downlink user in the transmission of the $ s $-th segment is given by
\begin{equation}\label{eqn:dl-rate}
R_{f,n,s} = N_{f,n,s} \mathbb{E}_{\mathbf{h}_{f,n,s}} \left[ \log_2 \left( 1 + \frac{||\mathbf{h}_{f,n,s}||^2 P_{f,n,s}}{N_T \sigma^2_z} \right) \right],
\end{equation}
where $ \sigma_z^2 $ is the power of noise as well as inter-cell interference, $ \mathbf{h}_{f,n,s} $ is the i.i.d. channel vector from the BS to the requesting user. Each element of $ \mathbf{h}_{f,n,s} $ is complex Gaussian distributed with zero mean and variance $ \rho_{f,n}\eta_{f,n,s} $. As a remark note that the transmission of one segment may consume a large number of frames, and the channel vector $ \mathbf{h}_{f,n,s} $ can be different from frame to frame. However, since we consider the ergodic channel capacity, the randomness in small-scale fading is averaged. Hence, the $ (f,n) $-th user can decode the $ s $-th segment only when $ R_{f,n,s}  \geq R_F/N_S. $ Similarly, the achievable data rate of the $ c $-th cache node is given by
\begin{equation}\label{eqn:dl-cache}
R_{f,n,s}^c = N_{f,n,s} \mathbb{E}_{\mathbf{h}_{f,n,s}^c} \left[ \log_2 \left( 1 + \frac{||\mathbf{h}_{f,n,s}^c||^2 P_{f,n,s}}{N_T \sigma^2_z} \right) \right],
\end{equation}
where $ \mathbf{h}_{f,n,s}^c $ is the i.i.d. channel vector from the BS to $ c $-th cache node. Each element of $ \mathbf{h}_{f,n,s}^c $ is complex Gaussian distributed with zero mean and variance $ \rho_{c}\eta_{f,n,s}^c $. The $ c $-th cache node can decode the $ s $-th segment only when $
R_{f,n,s}^c  \geq R_F/N_S. $

Moreover, let $ \eta^c_{k} $ be the shadowing attenuation in the $ k $-th proactive transmission opportunities of the BS to the $ c $-th cache node, $ P_k $ and $ N_k $ be the corresponding downlink transmission power and transmission symbols. The throughput achieved by the $ c $-th cache node is given by
\begin{equation}\label{eqn:proactive}
R_{k}^c = N_{k} \mathbb{E}_{\mathbf{h}_{k}^c} \left[ \log_2 \left( 1 + \frac{||\mathbf{h}_{k}^c||^2 P_{k}}{N_T \sigma^2_z} \right) \right],
\end{equation}
where $ \mathbf{h}_{k}^c $ is the i.i.d. channel vector from the BS to $ c $-th cache node. Each element of $ \mathbf{h}_{k}^c $ is complex Gaussian distributed with zero mean and variance $ \rho_{c}\eta_{k}^c $. The file segment can be decoded when $
R_{k}^c  \geq R_F/N_S. $ 

There is an scheduling problem in the downlink transmission: if more power or time resource is spent in downlink file transmission, more cache nodes are able to buffer the file (or some segments of the file), which may save the resource of the BS in the successive transmissions of the same file. Thus it is necessary to optimize the resource allocation during the whole lifetime of one file, instead of single transmission. This is a multi-stage optimization problem, which shall be addressed by the following two steps in this paper: 
\begin{itemize}
	\item Without proactive caching placement of the BS, we shall first formulate the downlink transmission as a random-stage MDP, and propose both optimal and low-complexity solutions, which is elaborated in Section \ref{sec:opt} and \ref{sec:approximation}.
	
	\item In Section \ref{sec:proactive}, proactive caching placement is considered to further suppress the transmission cost of BS.
\end{itemize}

\section{Optimal Control without Proactive Caching Placement} \label{sec:opt}

\subsection{Problem Formulation} 
Without proactive caching placement, the system state and scheduling policy are defined below. 

\begin{Definition}[System State]
	Before the transmission to the $ (f,n) $-th user, the system status is uniquely specified by $ S_{f,n} = \left[\mathcal{B}_{f,s}^c,\rho_{f,n} , \eta_{f,n,s}, \eta_{f,n,s}^c| \forall c=1,...,N_C; s=1,...,N_S  \right]$, where $ \mathcal{B}_{f,s}^c =1 $ means that the $ s $-th segment of the $ f $-th file has been successfully decoded by the $ c $-th cache node and $ \mathcal{B}_{f,s}^c =0 $ means otherwise.
\end{Definition}

\begin{Definition}[Downlink Multicast Policy]
	Suppose the $ s $-th segment is transmitted to the $ (f,n) $-th user via downlink. Given the system state $ S_{f,n} $, the scheduling policy $ \Omega_{f,n} $ ($ \forall f,n $) is a mapping from system state and the remaining lifetime $ T_{f,n} $ to the scheduling parameters $ P_{f,n,s} $ and $ N_{f,n,s} $ ($ \forall s $). Thus $$\Omega_{f,n}( S_{f,n}, T_{f,n})= \{(P_{f,n,s}, N_{f,n,s}) | \forall s\}. $$ Moreover, to guarantee that the requesting user can successfully decode the downlink data, the constraint $ R_{f,n,s} \geq R_F/N_S$ ($ \forall s $) should be satisfied.
\end{Definition}

Let $ \mathcal{C}_{f,n}^s = \cup_{\forall i, \mathcal{B}_{f,s}^i=1} \mathcal{C}_i $ be the area where the reqeusting users is able to receive the $ (f,s) $-th file segment from cache nodes, and $ \mathbf{l}_{f,n} $ be the location of the $ (f,n) $-th user. We use the following cost function to measure the weighted sum of energy and transmission time of the BS, which is spent on the $ (f,n) $-th user for the $ s $-th segment.
\begin{equation}\nonumber
g_{f,n,s} (P_{f,n,s}, N_{f,n,s}) = I(\mathbf{l}_{f,n} \notin \mathcal{C}_{f,n}^s) \times (w_e P_{f,n,s}N_{f,n,s} + w_t N_{f,n,s}),
\end{equation}
where $ w_e $ and $ w_t $ are the weights on transmission energy and transmission opportunities respectively, and $ I(\cdot) $ is the indicator function. Since transmission from cache nodes to users is made via other air interfaces, the transmission resource used by the cache nodes is not counted in this cost function. Hence the average cost spent on the overall lifetime of the $ f $-th file is given by
\begin{equation}\nonumber
\overline{g}_{f} \left( \{\Omega_{f,n}|\forall n\} \right) = \sum_N \mathbb{E}_{\eta,\rho} \left[ \frac{(\lambda_f T)^N}{N!} e^{-\lambda_f T} \sum_{n=1}^{N} \sum_{s=1}^{N_S} g_{f,n,s}\right],
\end{equation}
where the expectation is taken over all possible large-scale channel fading (including the shadowing effect $ \eta $ and requesting users' pathloss $ \rho $) in the system. The summation on $ N $ is due to the random number of requests as elaborated in (\ref{eqn:request}). As a result, the overall system cost function is given by
\begin{equation} \nonumber
\overline{G}(\{\Omega_{f,n}|\forall f,n\}) = \lim_{F \rightarrow +\infty} \frac{1}{F}\sum_{f=1}^{F} \overline{g}_{f} \left( \{\Omega_{f,n}|\forall n\} \right),
\end{equation}
and the system optimization problem can be written as

\begin{Problem}[Overall System Optimization] \label{prob:overall}
\begin{eqnarray}
&\min\limits_{\{\Omega_{f,n}|\forall f,n\}} &\overline{G}(\{\Omega_{f,n}|\forall f,n\}) \nonumber\\
&s.t.& R_{f,n,s}  \geq R_F/N_S,\  \forall f,n,s. \nonumber
\end{eqnarray}
\end{Problem}

Since there is no constraint on the cache nodes' buffer capacity and transmission resources, the above optimization problem can be further decoupled into the following sub-problems with respect to each file.
\begin{Problem}[Optimization on the $ f $-th File]\label{prob:main}
\begin{eqnarray}
&\min\limits_{\{\Omega_{f,n}|\forall n\}} &\overline{g}_f(\{\Omega_{f,n}|\forall n\}) \nonumber\\
&s.t.& R_{f,n,s}  \geq R_F/N_S, \ \forall n,s. \label{eqn:cont}
\end{eqnarray}
\end{Problem}

\subsection{Optimal Control Policy} \label{sec:optimal}

Note that Problem \ref{prob:main} is a dynamic programming problem with random number of stages, which cannot be solved by the standard approach as \cite{Bertsekas2000Dynamic}. We shall shown  in this section that the optimal solution of Problem \ref{prob:main} (i.e. policy iteration) can be obtained by applying value iteration on another MDP problem (Problem \ref{prob:fix} as follows)  with finite and fixed number of stages first, then solving a revised version of Bellman's equation. First of all, we introduce the following MDP problem with fixed number of stages.

\begin{Problem}[Optimization with Fixed Stage Number] \label{prob:fix}
	\begin{eqnarray}
	&\min\limits_{\{\Omega_{f,n}|\forall n\}} &\mathbb{E}_{\eta,\rho} [\sum_{n=1}^{N_R}\sum_{s=1}^{N_S} g_{f,n,s}]\nonumber\\
	&s.t.& R_{f,n,s}  \geq R_F/N_S, \ \forall n,s. \nonumber
	\end{eqnarray}
	where $ N_R $ is the fixed number of requests on the $ f $-th file.
\end{Problem}

As introduced in \cite{MDP1}, there is standard solution for the MDP problem with finite and fixed number of stages. The optimal solution of Problem \ref{prob:fix} can be deduced via the following Bellman's equation  
\begin{eqnarray}
V_{N_R-n+1}(S_{f,n})\!=\!\!\min_{\Omega_{f,n}(S_{f,n})}\! \!\bigg\{\! \sum_{s} g_{f,n,s}(\Omega_{f,n}) \!+\!\sum\limits_{S_{f,n+1}}{\!\!V_{N_R-n}(S_{f,n+1})Pr(S_{f,n+1}|S_{f,n},\Omega_{f,n})} \!\bigg\}, \forall S_{f,n},\label{eqn:bellman-fix}
\end{eqnarray}
where $ V_{N_R-n+1}(S_{f,n}) $ is usually named as value function of the $ n $-th stage, and $ S_{f,n+1} $ denotes the next state of the $ f $-th file given the current state $ S_{f,n} $. According to the definition of value function in finite-stage dynamic programming, $ V_{N_R-n+1}(S_{f,n}) $ is the average remaining cost of the $ f $-th file from the $ n $-th transmission to the $ N_R $-th transmission, given the system state of the $ n $-th stage $ S_{f,n} $.

Note that the large-scale fading is i.i.d. in each file transmission, the expectation on large-scale fading can be taken on both side of the above equation. Hence we have the following conclusion.

\begin{Lemma}[Bellman's Equation with Reduced Space]\label{Lemma:reduce_space}
The optimal control policy of Problem \ref{prob:fix} is the solution of the Bellman's equation with reduced state space as follows
\begin{eqnarray}\label{eqn:bellman-reduce}
\widetilde{V}_{N_R-n+1}(\widetilde{S}_{f,n})\!=\!\!\!\!\!\min_{\Omega_{f,n}(\widetilde{S}_{f,n})} \!\! \!\!\!\mathbb{E}_{\eta,\rho} \bigg\{\! \sum_{s}\! g_{f,n,s}(\Omega_{f,n}) \!+\!\!\!\sum\limits_{\widetilde{S}_{f,n+1}}{\!\!\widetilde{V}_{N_R-n}(\widetilde{S}_{f,n+1})Pr(\widetilde{S}_{f,n+1}|{S}_{f,n},\Omega_{f,n})} \! \bigg\},
\end{eqnarray}
where $\widetilde{S}_{f,n}= \{\mathcal{B}_{f,s}^c \in S_{f,n}|\forall c,s \}$, $\widetilde{V}_{N_R-n}(\widetilde{S}_{f,n})=\mathbb{E}_{\eta,\rho}[V_{N_R-n}(S_{f,n+1})] $, and $ \Omega_{f,n}(\widetilde{S}_{f,n}) = \{\Omega_{f,n}(S_{f,n})|\forall \rho_{f,n} , \eta_{f,n,s}, \eta_{f,n,s}^c \} $.
\end{Lemma}

\begin{proof}
Please refer to Appendix A.
\end{proof}

The standard iterative value iteration can be used to solve the Bellman's equation (\ref{eqn:bellman-reduce}), and obtain the value functions $ \widetilde{V}_{N_R-n+1}(\widetilde{S}_{f,n}) $ ($ \forall n, \widetilde{S}_{f,n} $). Since this is not the contribution of this paper, the detail procedure of value iteration is omitted here due to limited space. In the following lemma, we show that the optimization of Problem \ref{prob:main} can also be solved given the value functions $ \widetilde{V}_{N_R-n+1}(\widetilde{S}_{f,n}) $ ($ \forall n, \widetilde{S}_{f,n} $). 

\begin{Lemma}[Optimal Control Policy of Problem \ref{prob:main}]\label{Lemma:random_stage} With the value function $ \widetilde{V}_{N_R-n+1}(\widetilde{S}_{f,n}) $ ($ \forall n $), the optimal control policy for Problem \ref{prob:main}, denoted as $ \Omega_{f,n}^*(\widetilde{S}_{f,n},T_{f,n}) $, can be calculated as
\begin{eqnarray}\label{eqn:bellman-random}
\Omega_{f,n}^*(\widetilde{S}_{f,n},T_{f,n})&=&\arg \min_{\Omega_{f,n}(\widetilde{S}_{f,n})} \mathbb{E}_{\eta,\rho} \bigg\{\sum_{s} g_{f,n,s}(\Omega_{f,n}) + \nonumber\\
&&\sum\limits_{N,\widetilde{S}_{f,n+1}}\frac{(\lambda_f T_{f,n})^N}{N!} e^{-\lambda_f T_{f,n}}{\widetilde{V}_{N}(\widetilde{S}_{f,n+1})Pr(\widetilde{S}_{f,n+1}|{S}_{f,n},\Omega_{f,n})} \bigg\},
\end{eqnarray}
where $ \widetilde{S}_{f,n+1} $ is the next state of the $ f $-th file given the current state $ S_{f,n} $.
\end{Lemma}
\begin{proof} 
Note that (1) $ \widetilde{V}_{N}(\widetilde{S}_{f,n+1}) $ denotes the averaged cost for $ N $ times of requests on the $ f $-th file, give the system state $ \widetilde{S}_{f,n+1} $; (2) $ \frac{(\lambda_f T_{f,n})^N}{N!} e^{-\lambda_f T_{f,n}} $ is the probability that there are $ N $ times of file requests within the duration $ T_{f,n} $. The conclusion is straightforward since the average cost of the BS, given the current state $ \widetilde{S}_{f,n} $ and remaining lifetime $ T_{f,n} $, is 
$$ \mathbb{E}_{\eta,\rho} \bigg\{\sum_{s} g_{f,n,s}(\Omega_{f,n}) + \sum\limits_{N,\widetilde{S}_{f,n+1}}\frac{(\lambda_f T_{f,n})^N}{N!} e^{-\lambda_f T_{f,n}}{\widetilde{V}_{N}(\widetilde{S}_{f,n+1})Pr(\widetilde{S}_{f,n+1}|{S}_{f,n},\Omega_{f,n})} \bigg\}. $$
\end{proof}

As a summary, the optimal control policy of the Problem \ref{prob:main} can be solved via the following two steps.
\begin{itemize}
	\item \textbf{Value Iteration}: Calculate the value function $ \widetilde{V}_{N_R-n+1}(\widetilde{S}_{f,n})( \forall n, \widetilde{S}_{f,n} ) $  via Bellman's equation with reduced space (\ref{eqn:bellman-reduce}).
	\item \textbf{Policy Iteration}: Obtain optimal control policy from (\ref{eqn:bellman-random}).
\end{itemize}
Note that the state space of the system is actually huge. For example, suppose there are 20 cache nodes in the system and 10 segments per file, the dimensional of system state is $2^{200}$, which is referred to as the {\em curse of dimensionality} in dynamic programming problem. Thus although the optimal solution algorithm can be provided, the computation complexity is intractable. Hence, we continue to propose a low-complexity solution based on the technique of approximated MDP and reinforcement learning in the following section.

\section{Low-Complexity Solution via Approximated MDP}\label{sec:approximation}

In this section, we shall introduce a novel linear approximation approach on the value function $ \widetilde{V}_{N_R-n+1}(\widetilde{S}_{f,n}) $, and elaborate the on-line control policy to determine the control actions given the current system state and approximated value function. In order to evaluate the approximated value function, we first provide analytical expressions given the distribution of requesting users (e.g. uniform distribution); then an online learning algorithm is introduced for unknown distribution statistics of requesting users.

\subsection{Approximation on Value Function} \label{sub:app}

We first define the notations for following reference system states.
\begin{itemize}
	\item $ \widetilde{S}_{f}^* $ is the state of $f$-th file where all the cache nodes have successfully decoded the whole file. Thus $
	\widetilde{S}_{f,n}^*= \{\mathcal{B}_{f,s}^{c} =1| \forall c, s\} .$
	
	\item $ \widetilde{S}_{f}^{i,s} $ is the state of $ f $-th file transmission where only the $ s $-th segment at the $ i$-th cache node is not successfully decoded. $
	\widetilde{S}_{f}^{i,s}= \{\mathcal{B}_{f,s}^{i} = 0, \mathcal{B}_{f,t}^{j}=1| \forall (j,t) \neq (i,s) \} .$
\end{itemize}
Hence, we approximate the value function $ \widetilde{V}_{N_R-n+1}(\widetilde{S}_{f,n}) $ linearly as 	\begin{equation}\label{eqn:app-v}
\widetilde{V}_{N_R-n+1}(\widetilde{S}_{f,n}) \approx \widetilde{V}_{N_R-n+1}(\widetilde{S}_{f}^*) + \sum_{\{(i,s)|\forall \mathcal{B}^i_{f,s}(\widetilde{S}_{f,n}) =0\}} \bigg( \widetilde{V}_{N_R-n+1}(\widetilde{S}_{f}^{i,s})-\widetilde{V}_{N_R-n+1}(\widetilde{S}_{f}^*)\bigg),
\end{equation}
where $\mathcal{B}^i_{f,s}(\widetilde{S}_{f,n})  $ means the parameter of $ \mathcal{B}^i_{f,s} $ in the system state $ \widetilde{S}_{f,n} $. An example of approximated value function is elaborated below.
\begin{figure}[tb]
	\centering
	\includegraphics[scale = 0.5]{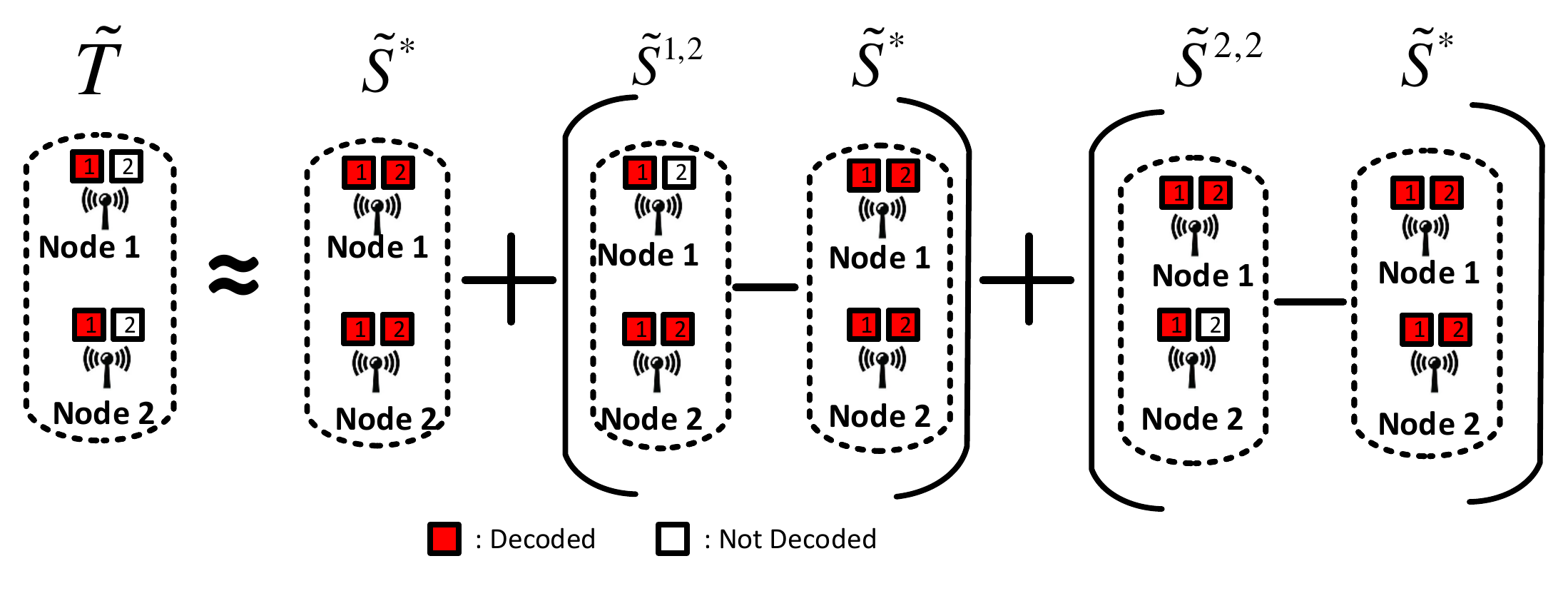}
	\caption{One example of approximated value function.}
	\label{fig:example2}
\end{figure}
\begin{Example}[Approximated Value Function]
	An example on the approximated value function is illustrated in Fig. \ref{fig:example2}, where there are two cache nodes and downlink file (say the $ f $-th file) is divided into two segments. For the system state $  \widetilde{T} = [\mathcal{B}^1_{f,1}, \mathcal{B}^1_{f,2}, \mathcal{B}^2_{f,1}, \mathcal{B}^2_{f,2}] = [1, 0, 1, 0] $, the value function on the $ n $-th stage can be approximated as
	\begin{eqnarray}
	\widetilde{V}_{N_R-n+1}(\widetilde{T}) &\approx& \widetilde{V}_{N_R-n+1}(\widetilde{S}_{f}^*) +\bigg( \widetilde{V}_{N_R-n+1}(\widetilde{S}_{f}^{1,2})-\widetilde{V}_{N_R-n+1}(\widetilde{S}_{f}^*)\bigg) \nonumber\\
	&&+\bigg( \widetilde{V}_{N_R-n+1}(\widetilde{S}_{f}^{2,2})-\widetilde{V}_{N_R-n+1}(\widetilde{S}_{f}^*)\bigg). \nonumber
	\end{eqnarray}
	In the right hand side of the above approximation, the first term counts the transmission cost for the users outside the coverage of cache nodes; the second term approximates the cost on the second segment transmission to the users within the coverage of the first cache node; and the third term approximates the cost on the second segment transmission to the users within the coverage of the second cache node. Note that there is no transmission cost on the first segment for the users within the coverage of both cache nodes.
\end{Example}

In order to apply this approximation on all value function, it is necessary to obtain the value of $\widetilde{V}_{N_R-n+1}(\widetilde{S}_{f}^*)$ and $\widetilde{V}_{N_R-n+1}(\widetilde{S}_{f}^{i,s})$ for all $ n,i $, and $ s $ via (\ref{eqn:bellman-reduce}). In the following, we provide the analytically expressions for them with the knowledge of the requesting users' distribution. Moreover, in the next section, an online learning algorithm is proposed for the evaluation of $\widetilde{V}_{N_R-n+1}(\widetilde{S}_{f}^*)$ and $\widetilde{V}_{N_R-n+1}(\widetilde{S}_{f}^{i,s})$ with unknown spatial distribution of requesting users.

\subsubsection{Evaluation of $ \widetilde{V}_{N_R-n+1}(\widetilde{S}_{f}^*) $}

Note that system state $ \widetilde{S}_{f}^* $ represents the situation that all the cache nodes have already decoded the $ f $-th file, the purpose of downlink transmission is only to make sure that the requesting users, which are outside of the coverage of cache nodes, can decode the downlink file. Hence it is clear that
\begin{eqnarray}\nonumber
\widetilde{V}_{N_R-n+1}(\widetilde{S}_{f}^*) &=& (N_R-n+1)\Pr(\mathbf{l}_{f,n} \notin \mathcal{C}_{f,n}^s) \nonumber\\
&&\mathbb{E}_{\rho,\eta} \bigg[\sum_s \min \limits_{P_{f,n,s} \atop N_{f,n,s}}  w_e P_{f,n,s}N_{f,n,s} + w_t N_{f,n,s}  \bigg| \mathbf{l}_{f,n} \notin \mathcal{C}_{f,n}^s\bigg] \nonumber \\
&s.t.& R_{f,n,s} \geq R_F/N_S, \ \forall s.\nonumber
\end{eqnarray}
The above value function can be calculated with analytical expression, which is elaborated below.
\begin{Lemma}\label{Lemma:Segment}
The value function $ \widetilde{V}_{N_R-n+1}(\widetilde{S}_{f}^*) $ is given by
\begin{eqnarray}\nonumber
\widetilde{V}_{N_R-n+1}(\widetilde{S}_{f}^*) 
\approx (N_R-n+1)\Pr(\mathbf{l}_{f,n} \notin \mathcal{C}_{f,n}^s)
\mathbb{E}_{\rho,\eta} \bigg[\sum_s (w_e P_{f,n,s}^*N_{r,f,s}^* + w_t N_{f,n,s}^* ) \bigg| \mathbf{l}_{f,n} \notin \mathcal{C}_{f,n}^s \bigg], \nonumber 
\end{eqnarray}
where $
P_{f,n,s}^* =\frac{w_t} {w_e \mathbb{W}(\frac{2^{\theta}w_t}{ew_e}) },
N_{f,n,s}^* =\frac {R_F}{N_S[\theta+\log_2(P_{f,n,s}^*)]} $, $
\theta =	\mathbb{E}_{\mathbf{h}_{f,n,s}} \left[ \log_2 \left(  \frac{||\mathbf{h}_{f,n,s}||^2}{N_T \sigma^2_z}\right) \right] $, 
and $\mathbb{W}(x)$ is the Lambert-W function \cite{W}.
\end{Lemma}

\begin{proof}
	Please refer to Appendix B.	
\end{proof}
\subsubsection{Evaluation of $\widetilde{V}_{N_R-n+1}(\widetilde{S}_{f}^{i,s})$}

Given system state $ \widetilde{S}_{f}^{i,s} $ for arbitrary stage, there are only two possible next system states $  \widetilde{S}_{f}^{i,s} $ and $  \widetilde{S}_{f}^{*} $, which are discussed below.
\begin{itemize}
	\item When $ \rho_{f,n}  \eta_{f,n,s} \leq \rho_i \eta_{f,n,s}^i $, thus $ R_{f,n,s} \leq R_{f,n,s}^i$, the $ i $-th cache node is alway able to decode the $ s $-th file segment give that the transmission constraint (\ref{eqn:cont}) should be satisfied. Thus the next state must be $  \widetilde{S}_{f}^{*} $. In this case, the optimized RHS of (\ref{eqn:bellman-reduce}) is given by
	\begin{eqnarray} \label{eqn:qc}
	Q_c (\widetilde{S}_f^{i,s})&=& \min\limits_{\Omega_{f,n}(\widetilde{S}_f^{i,s})}  \sum_{t} g_{f,n,t}(\widetilde{S}_f^{i,s},\Omega_{f,n}) +{\widetilde{V}_{N_R-n}(\widetilde{S}_{f}^{*})}, \\
	&s.t.& R_{f,n,t} = \frac{R_F}{N_S},\ \  \forall t. \nonumber
	\end{eqnarray}
	
	\item When $ \rho_{f,n}  \eta_{f,n,s} > \rho_i \eta_{f,n,s}^i $, thus $ R_{f,n,s} > R_{f,n,s}^i$, the BS can choose to secure the transmission of the $ s $-th segment to the $ (f,n) $-th user or the $ i $-th cache node. Hence the optimized  RHS of (\ref{eqn:bellman-reduce}) is given by $$
	\mathbb{E}  \bigg\{ \min \bigg[ Q_u(\widetilde{S}_f^{i,s}),Q_{i,s}(\widetilde{S}_f^{i,s}) \bigg] \bigg\},$$ 
	where $Q_u$ and $Q_{i,s}$ are defined below.
	\begin{eqnarray}\label{eqn:a_u}
	Q_u(\widetilde{S}_f^{i,s}) & = & \min\limits_{\Omega_f(\widetilde{S}_f^{i,s})} \sum_{t} g_{f,n,t}(\widetilde{S}_f^{i,s},\Omega_{f,n}) +{\widetilde{V}_{N_R-n}(\widetilde{S}_{f}^{i,s})},\ \ s.t. \ \ R_{f,n,t}=\frac{R_F}{N_S}, \  \forall t
	\end{eqnarray}
	\begin{eqnarray}\label{eqn:a_i}
	Q_{i,s}(\widetilde{S}_f^{i,s})&=&\min\limits_{\Omega_f(\widetilde{S}_f^{i,s})} \sum_{t} g_{f,n,t}(\widetilde{S}_f^{i,s},\Omega_f) + {\widetilde{V}_{N_R-n}(\widetilde{S}_{f}^{*})}\\
	&s.t. & R_{f,n,s}^i \!=\!\frac{ R_F}{N_S} \mbox{ and } R_{f,n,t}\!=\!\frac {R_F}{N_S}, \!  \forall t \neq s. \nonumber
	\end{eqnarray}
\end{itemize}
As a result, the expression of $\widetilde{V}_{N_R-n+1}(\widetilde{S}_{f}^{i,s})$ is summarized by the following lemma.
\begin{Lemma}\label{Lemma:default_one_Segment}
	 The value function  $\widetilde{V}_{N_R-n+1}(\widetilde{S}_{f}^{i,s})$ is given below
	 \begin{eqnarray}\label{eqn:v_default one fragment }
	 \widetilde{V}_{N_R-n+1}(\widetilde{S}_{f}^{i,s})&=&\mathbb{E}_{\eta,\rho}[Q_c(\widetilde{S}_f^{i,s})|R_{f,n,s} \leq R_{f,n,s}^i]Pr(R_{f,n,s} \leq R_{f,n,s}^i) + \nonumber\\
	 &&\mathbb{E}_{\eta,\rho}[\min \{Q_u(\widetilde{S}_f^{i,s}), Q_{i,s}(\widetilde{S}_f^{i,s})\}|R_{f,n,s} > R_{f,n,s}^i]Pr(R_{f,n,s}> R_{f,n,s}^i).
	 \end{eqnarray}
	 Moreover, the optimal control actions for $ Q_c $ and $ Q_u $ are the same as Lemma \ref{Lemma:Segment}. The optimal control action for  $ Q_{i,s} $ is given by
	 \begin{equation*}
	 P_{f,n,s}^*=\frac{w_t} {w_e \mathbb{W}(\frac{2^{\theta^i}w_t}{ew_e}) },N_{f,n,s}^*=\frac {R_F}{N_S[\theta^i+\log_2(P_{f,n,s})]}, \theta^i=	\mathbb{E}_{\mathbf{h}_{f,n,s}} \left[ \log_2 \left(  \frac{||\mathbf{h}_{f,n,s}^i||^2}{N_T \sigma^2_z}\right) \right]
	 \end{equation*}
	 and $ \forall t\neq s $
	 \begin{equation*}
	 P_{f,n,t}^* =\frac{w_t} {w_e \mathbb{W}(\frac{2^{\theta}w_t}{ew_e}) },
	 N_{f,n,t}^* =\frac {R_F}{N_S[\theta+\log_2(P_{f,n,s}^*)]}, 
	 \theta =	\mathbb{E}_{\mathbf{h}_{f,n,t}} \left[ \log_2 \left(  \frac{||\mathbf{h}_{f,n,t}||^2}{N_T \sigma^2_z}\right) \right].
	 \end{equation*}
\end{Lemma}

\begin{proof}
The proof is similar to that of Lemma \ref{Lemma:Segment}, and it is omitted here.	
\end{proof}

With the distribution knowledge of large-scale fading, the value functions $\widetilde{V}_{N_R-n+1}(\widetilde{S}_{f}^*)$ and $\widetilde{V}_{N_R-n+1}(\widetilde{S}_{f}^{i,s})$ can be calculated according to above analytical expressions. Note that these value functions are the same for different files, they can be applied on all the files' transmission once calculated. Thus the complexity is low.

\subsection{Low-Complexity Downlink Multicast Policy } \label{sub:low}

With the value function $\widetilde{V}_{N_R-n+1}(\widetilde{S}_{f}^*)$ and $\widetilde{V}_{N_R-n+1}(\widetilde{S}_{f}^{i,s})$, the value function for arbitrary system state in arbitrary transmission stage can be approximated via (\ref{eqn:app-v}). Hence the online control action for arbitrary system state $ S_{f,n} $, denoted as $ \Omega_{f,n}^{*}({S}_{f,n}) $, can be obtained by 
the following optimization problem.

\begin{Problem}[Online Optimization]\label{prob:online}
\begin{eqnarray}
\Omega_{f,n}^{*}({S}_{f,n})&=\arg\min& \sum_s  g_{f,n,s}({S}_{f,n},\Omega_{f,n})  +\sum\limits_{N}\frac{(\lambda_f T_{f,n})^N}{N!} e^{-\lambda_f T_{f,n}}{\widetilde{V}_{N}(\widetilde{S}_{f,n+1})}\nonumber\\
&s.t.& R_{f,n,s} \geq R_F / N_S, \ \forall s. \nonumber
\end{eqnarray}

\end{Problem}

Since the value function $ \widetilde{V}_{N}(\widetilde{S}_{f,n+1}) $ is approximated by (\ref{eqn:app-v}), the optimization in Problem \ref{prob:online} can be further decoupled for each segment. For the $ s $-th segment ($ \forall s $), the solution of Problem \ref{prob:online} can be obtained by the following problem, given that the requesting user cannot find the segment from nearby cache nodes.
\begin{Problem}[Online Optimization for the $ s $-th Segment]\label{prob:online-s}
	\begin{eqnarray}
	&&\{P_{f,n,s}^{*}, N_{f,n,s}^{*}\}\nonumber\\
	&=\arg\min& g_{f,n,s}({S}_{f,n},\Omega_{f,n})  + \sum\limits_{N}\frac{(\lambda_f T_{f,n})^N}{N!} e^{-\lambda_f T_{f,n}} \sum\limits_{\{i|\forall \mathcal{B}^i_{f,s}(\widetilde{S}_{f,n+1})=0\}} \widetilde{V}_{N}(\widetilde{S}_{f}^{i,s}) - \widetilde{V}_{N}(\widetilde{S}_{f}^{*})\nonumber\\
	&s.t.& R_{f,n,s} \geq R_F / N_S, \nonumber
	\end{eqnarray}
where $ \mathcal{B}^i_{f,s}(\widetilde{S}_{f,n+1}) $ represents the buffer status for the $ (f,s) $-th segment in the $ i $-th cache node.
\end{Problem}

Due to the second term of objective in Problem \ref{prob:online-s}, the BS should first choose the cache nodes for downlink receiving, in addition to the requesting user. Based on the selection, the optimal power and transmission time can be derived. This is an integrated continuous and discrete optimization, its solution algorithm is summarized below.

\begin{Algorithm}[Online Scheduling with Approximated Value Function]\label{alg:AMDP}
	
Given the system state $ S_{f,n} $, let $ d_1,d_2,.. $ be the indexes of cache nodes, whose large-scale attenuation to the BS in the $ s $-th segment is worse than the $ (f,n) $-th user. Moreover, without loss of generality, it is assumed that $ \rho_{d_1}\eta_{f,n,s}^{d_1} \leq  \rho_{d_2}\eta_{f,n,s}^{d_2} \leq ... \leq \rho_{f,n,s}\eta_{f,n,s}$. The optimal control action for the $ s $-th segment ($ \forall s $) can be obtained below.
\begin{itemize}
	\item For each $ i $, suppose the $ d_i $-th cache node are involved for downlink receiving, the optimal power and transmission time control is given by
	\begin{align}
	Q_{d_i,s}^{*}({S}_{f,n})&=\min\limits_{P_{f,n,s} \atop N_{f,n,s}} g_{f,n,s}({S}_{f,n},\Omega_{f,n})  + \sum\limits_{N}\frac{(\lambda_f T_{f,n})^N}{N!} e^{-\lambda_f T_{f,n}} \bigg[\sum\limits_{j=\{d_1,...,d_{i-1}\}} \!\!\!\!\!\!\widetilde{V}_{N}(\widetilde{S}_{f}^{j,s}) - \widetilde{V}_{N}(\widetilde{S}_{f}^{*})\bigg]\nonumber\\
	&s.t. \quad R_{f,n,s}^{d_i} = R_F/N_S. \nonumber
	\end{align}
	The optimal solution, denoted as $[P_{f,n,s}^{d_i},N_{f,n,s}^{d_i}]$, can be derived similar to Lemma \ref{Lemma:Segment}. 
	\item Let $
d^{*}=\arg \min\limits_{d_i} Q_{d_i,s}^* $, the solution of Problem \ref{prob:online-s} is then given by $
[P_{f,n,s}^{*},N_{f,n,s}^{*}]=[P_{f,n,s}^{d^*},N_{f,n,s}^{d^*}]. $
	
\end{itemize}

\end{Algorithm}

\subsection{Learning Algorithm for Approximated Value Function}\label{sec:learning}

In Section \ref{sub:app}, the values of $\widetilde{V}_{N_R-n+1}(\widetilde{S}_{f}^*)$ and $\widetilde{V}_{N_R-n+1}(\widetilde{S}_{f}^{i,s})$ are evaluated analytically by assuming the available distribution knowledge of requesting users. However in practice, the statistics of the distribution may be unknown to the BS. In order to address this issue, we introduce in the following a learning-based online algorithm to evaluate the value functions  $\widetilde{V}_{N_R-n+1}(\widetilde{S}_{f}^*)$ and $\widetilde{V}_{N_R-n+1}(\widetilde{S}_{f}^{i,s})$.

\begin{Algorithm}[Reinforcement Learning for Value Function]\label{alg:learning}
	Suppose requesting users for all files are distributed i.i.d. according to certain spatial distribution with unknown statistics. When one user becomes active to request one file, its location and channel information can be used to update all the value functions. Note that $\widetilde{V}_{N_R-n+1}(\widetilde{S}_{f}^*)$ and $\widetilde{V}_{N_R-n+1}(\widetilde{S}_{f}^{i,s})$ ($ \forall n,i,s $) are identical for all files, we only elaborate the update on the $ f $-th file's value functions. Specifically, the learning procedure is described below.
	\begin{itemize}
		\item {\bf Step 1}: Let $ t=0 $. Initialize the value of $\widetilde{V}_{N_R-n+1}(\widetilde{S}_{f}^*)$ and $\widetilde{V}_{N_R-n+1}(\widetilde{S}_{f}^{i,s})$ ($ \forall n,i,s $), and denote them as $\widetilde{V}^t_{N_R-n+1}(\widetilde{S}_{f}^*)$ and $\widetilde{V}^t_{N_R-n+1}(\widetilde{S}_{f}^{i,s})$. This initialization can be done by assuming all the users appear uniformly in the cell coverage, hence the approach in Section \ref{sub:app} can be applied.
		
		\item {\bf Step 2}: $ t=t+1 $.
		
		\item {\bf Step 3}: When one new user appears to request any downlink file, update the value functions with its location and channel information. Suppose it is the $ m $-th request on the $ g $-th file, its location is defined as $ \mathbf{l}_{m,g} $, we have
		\begin{eqnarray}
		\widetilde{V}_{N_R-n+1}^t(\widetilde{S}_{f}^*)&=&\frac{t}{t+1}\widetilde{V}_{N_R-n+1}^{t-1}(\widetilde{S}_{f}^*)+\nonumber\\
		&&\frac{1}{t+1}(N_R-n+1)I(\mathbf{l}_{g,m} \notin \mathcal{C})
		\sum_s (w_e P_{g,m,s}^*N_{g,m,s}^* + w_t N_{g,m,s}^*),\nonumber
		\end{eqnarray}
		where $P_{g,m,s}^* =\frac{w_t} {w_e \mathbb{W}(\frac{2^{\theta_t}w_t}{ew_e}) },
		N_{g,m,s}^* =\frac {R_F}{N_S[\theta_t+\log_2(P_{f,n,s}^*)]} $, $
		\theta_t =	\mathbb{E}_{\mathbf{h}_{g,m,s}} \left[ \log_2 \left(  \frac{||\mathbf{h}_{g,m,s}||^2}{N_T \sigma^2_z}\right) \right] $.
		\begin{eqnarray}
		\widetilde{V}_{N_R-n+1}^t(\widetilde{S}_{f}^{i,s})&=&\frac{t}{t+1}\widetilde{V}_{N_R-n+1}^{t-1}(\widetilde{S}_{f}^{i,s})+\frac{1}{t+1}\bigg\{ Q_c(\widetilde{S}_{g}^{i,s}) I(R_{g,m,s} \leq R_{g,m,s}^i)\nonumber\\
		&& + \min \{Q_u (\widetilde{S}_{g}^{i,s}), Q_{i,s} (\widetilde{S}_{g}^{i,s})\} I(R_{g,m,s}> R_{g,m,s}^i)\bigg\}, \nonumber
		\end{eqnarray}
		where $ Q_c $, $Q_u$ and $Q_{i,s}$ are defined in (\ref{eqn:qc}), 
		(\ref{eqn:a_u}) and (\ref{eqn:a_i}) respectively.
		
		\item Step 4: If $\max\{|\widetilde{V}_{N_R-n+1}^t(\widetilde{S}_{f}^{i,s})- \widetilde{V}_{N_R-n+1}^{t-1}(\widetilde{S}_{f}^{i,s})|, |\widetilde{V}_{N_R-n+1}^t(\widetilde{S}_{f}^{*})- \widetilde{V}_{N_R-n+1}^{t-1}(\widetilde{S}_{f}^{*})| \big| \forall n,i,s\}  $ is greater than one threshold $ \tau $, the algorithm goes to Step 2; otherwise, the algorithm terminates.
	\end{itemize}
\end{Algorithm}

Moreover, we have the following lemma on the convergence of above iterative learning algorithm.

\begin{Lemma}\label{Lemma:learning}
	The Algorithm \ref{alg:learning} will converge to the true value of $\widetilde{V}_{N_R-n+1}(\widetilde{S}_{f}^*)$ and $\widetilde{V}_{N_R-n+1}(\widetilde{S}_{f}^{i,s})$ ($ \forall f,n,i,s $). Thus
	\begin{equation}
	\lim\limits_{t \rightarrow +\infty } \widetilde{V}^t_{N_R-n+1}(\widetilde{S}_{f}^*) = \widetilde{V}_{N_R-n+1}(\widetilde{S}_{f}^*) \nonumber,
	\end{equation}
	\begin{equation}
	\lim\limits_{t \rightarrow +\infty } \widetilde{V}^t_{N_R-n+1}(\widetilde{S}_{f}^{i,s}) = \widetilde{V}_{N_R-n+1}(\widetilde{S}_{f}^{i,s})\nonumber.
	\end{equation}
\end{Lemma}

\begin{proof}
	Please refer to Appendix C. 
\end{proof}

\subsection{Bounds on Approximated Value function}

Since the control policy is derived according to approximated value function in Section \ref{sub:low}, it is necessary to know the bound on the approximation error. Hence we introduce the following conclusion.

\begin{Lemma}[Bounds of Value Functions] \label{lem:bound} With arbitrary spatial distribution of requesting users, the upper-bound in (\ref{eqn:upper}) holds for value function $ \widetilde{V}_{N_R-n+1}(\widetilde{S}_{f,n}) $ ($ \forall n $). Moreover, if there is no overlap in the service region of cache nodes, the lower-bound in (\ref{eqn:lower}) also holds for  $ \widetilde{V}_{N_R-n+1}(\widetilde{S}_{f,n}) $.

\begin{equation}\label{eqn:upper}
\widetilde{V}_{N_R-n+1}(\widetilde{S}_{f,n}) \leq  \widetilde{V}_{N_R-n+1}(\widetilde{S}_{f,n}^*)+ \sum_{\{(i,s)|\forall \mathbf{B}^i_{f,s}(\widetilde{S}_{f,n}) = 0\}}  \bigg( \widetilde{V}_{N_R-n+1}(\widetilde{S}_{f}^{i,s})-\widetilde{V}_{N_R-n+1}(\widetilde{S}_{f,n}^*)\bigg)
\end{equation}
\begin{equation}\label{eqn:lower}
\widetilde{V}_{N_R-n+1}(\widetilde{S}_{f,n}) \geq  \widetilde{V}_{N_R-n+1}(\widetilde{S}_{f,n}^*)+ \sum_{{\{(i,s)|\forall \mathbf{B}^i_{f,s}(\widetilde{S}_{f,n}) = 0\}}} \bigg( \widetilde{V}_{1}(\widetilde{S}_{f}^{i,s})-\widetilde{V}_{1}(\widetilde{S}_{f,n}^*)\bigg)
\end{equation}	
\end{Lemma}
\begin{proof}
	Please refer to Appendix D.
\end{proof}

Notice that the proposed linear approximation on value function is actually the upper-bound in (\ref{eqn:upper}), the gap between the approximated value function and the actual value function $ \widetilde{V}_{N_R-n+1}(\widetilde{S}_{f,n}) $, denoted as $ \mathcal{E}_{N_R-n+1}(\widetilde{S}_{f,n}) $, is given by
\begin{eqnarray} 
\mathcal{E}_{N_R-n+1}(\widetilde{S}_{f,n}) \leq  \sum_{\{(i,s)|\forall \mathcal{B}^i_{f,s}(\widetilde{S}_{f,n}) = 0\}} 
\bigg\{\widetilde{V}_{N_R-n+1}(\widetilde{S}_{f}^{i,s})-\widetilde{V}_{N_R-n+1}(\widetilde{S}_{f}^*)	-\widetilde{V}_{1}(\widetilde{S}_{f}^{i,s})+\widetilde{V}_{1}(\widetilde{S}_{f}^*)\bigg\}.\nonumber
\end{eqnarray}
According to the definition of value function, the average system cost on the $ f $-th file with optimal control can be written as
\begin{equation}
\overline{g}_f^{*} = \sum_{N_R}\frac{(\lambda_f T)^{N_R}}{N_R!} e^{-\lambda_f T} \widetilde{V}_{N_R} (\widetilde{S}_f^0), \label{eqn:system}
\end{equation}
where $ \widetilde{S}_f^0 $ denotes the system state with empty buffer in all cache nodes. Hence the upper and lower bounds on $ \overline{g}_f^{*} $ can be obtained by applying (\ref{eqn:upper}) and (\ref{eqn:lower}) on $ \widetilde{V}_{N_R} (\widetilde{S}_f^0) $ in (\ref{eqn:system}).

\section{Proactive Caching Placement}\label{sec:proactive}

In the above sections, the downlink transmission happens only when users raise requests on file transmission. This cannot exploit the channel temporal diversity on the shadowing effect. For example, when the channels from the BS to some cache nodes are good, the downlink resource consumption on caching placement can be low, and the proactive file transmission to these cache nodes may further save the overall downlink cost. Hence, in this section, we consider the scenario that the BS can deliver one file segment to cache nodes per every $ T_p $ seconds. It is assumed that the transmission of one file segment can be completed within the coherent time of shadowing attenuation. Thus we define the following control policy.

 \begin{Definition}[Proactive Caching Placement Policy]
	In the $ k $-th proactive transmission opportunity, given the buffer status of each cache node $ \{ \mathcal{B}_{f,s}^c|\forall c, f, s \} $, the shadowing from the BS to each cache nodes $ \{\eta^c_k | \forall c \} $, and the remaining lifetime of each file $ \{ T^k_f| \forall f \} $, the BS should determine the selected file segment $ (F_k, S_k) $ and the downlink transmission parameters $ P_{k} $ and $ N_{k} $ for the selected $ (F_k, S_k) $-th file segment. Thus denote $ S^k =\left[\{ \mathcal{B}_{f,s}^c|\forall c, f, s \},\{\eta^c_k | \forall c \}, \{ T^k_f| \forall f \} \right]$, the placement policy can be written as $$\Omega_{k}( S^k)= [F_k, S_k, P_{k}, N_{k}]. $$
\end{Definition}

Suppose before the $ k $-th proactive transmission opportunity, the buffer status of cache nodes is $ \widetilde{S}_{f}^k = \{ \mathcal{B}^c_{f,s} |\forall c,f,s \} $. The average cost of the BS for the $ f $-th file without proactive caching placement is given by
\begin{eqnarray}
\overline{g}_f^k(\widetilde{S}_{f}^k, T_{f}^k) = \sum\limits_{N}\frac{(\lambda_f T_{f}^k)^N}{N!} e^{-\lambda_f T_{f}^k}{\widetilde{V}_{N}(\widetilde{S}_{f}^k)} \nonumber,
\end{eqnarray}
where $ \widetilde{V}_{N}(\widetilde{S}_{f}^k) $ can be approximated linearly via (\ref{eqn:app-v}). In each proactive caching placement opportunity, the BS can choose one file segment which is expected to mostly suppress the expected cost of the BS. As a result, we have the following algorithm on the proactive caching placement.

\begin{Algorithm}[Proactive Caching Placement] On each proactive transmission opportunity (say the $ k $-th opportunity), the algorithm to determine the proactive caching placement policy $ \Omega_k $ is elaborated below.
	
\begin{itemize}
\item Step 1: For each file segment (say the $ (f,s) $-th one), evaluate the expected cost saving if the BS choose deliver this segment proactively by solving the following Problem \ref{prob:proactive}.

\item Step 2: The $ (F_k,S_k) $-th segment is chosen when the following two conditions are satisfied: 
\begin{itemize}
	\item $ (F_k,S_k)  = \arg\max\limits_{(f,s)} \Delta g_{f,s}^k $, where $ \Delta g_{f,s}^k $ is defined in Problem \ref{prob:proactive}.
	\item $ \Delta g_{F_k,S_k}^k \geq \tau^{'} $, where $ \tau^{'} > 1 $ is certain threshold.
\end{itemize}
\end{itemize}
\end{Algorithm}

\begin{Problem}[Proactive placement for $ (f,s) $-th Segment] \label{prob:proactive}	
\begin{eqnarray}
	\Delta g_{f,s}^k &=& \max\limits_{P_{f,s}^k, N_{f,s}^k, \mathcal{A}_f^k}	\overline{g}_f^k(\widetilde{S}_{f}^k, T_{f}^k) /  \left[w_e P_{f,s}^k N_{f,s}^k + w_t  N_{f,s}^k + \overline{g}_f^k(\breve{S}_{f}^k(\mathcal{A}_f^k), T_{f}^k)\right] \nonumber\\
	&s.t.& R_k^c \geq R_F/N_S, \ \forall c \in \mathcal{A}_f^k, \nonumber
\end{eqnarray}
where $ P_{f,s}^k, N_{f,s}^k$ and $ \mathcal{A}_f^k $ represents the transmission power, transmission time and the set of multicast cache nodes, $ \breve{S}_{f}^k(\mathcal{A}_f^k) $ denotes the cache buffer status where the cache nodes in $ \mathcal{A}_f^k $ have successfully decoded the $ (f,s) $-th segment given the previous status $ \widetilde{S}_{f}^k $.
\end{Problem}

\begin{Algorithm}[Solution Algorithm for Problem \ref{prob:proactive}] Given the system state $ \widetilde{S}_{f}^k $ and the remaining lifetime $T_f^k$, let $ d_1,d_2,... $($\forall d_i,\mathcal{B}^{d_i}_{f,s}(\widetilde{S}_{f}^{k} )=0 $) be the indexes of cache nodes satisfying $ \rho_{d_1}\eta_{k}^{d_1} \leq  \rho_{d_2}\eta_{k}^{d_2} \leq \cdots $. Then $\Delta g_{f,s}^k$ can be obtained below.
   \begin{itemize}
   	\item Step 1: For each $i$, suppose the $d_i$-th cache node are involved for downlink receiving, the optimal
   	power and transmission time control is given by
   		\begin{equation}
   		Q_{d_i,s}^{*} =\min\limits_{P_{f,s}^{k,d_i} \atop N_{f,s}^{k,d_i}}  w_e P_{f,s}^{k,d_i} N_{f,s}^{k,d_i} + w_t  N_{f,s}^{k,d_i}  + \sum\limits_{N}\frac{(\lambda_f T_{f}^k)^N}{N!} e^{-\lambda_f T_{f}^k} \bigg[\sum\limits_{j=\{d_1,...,d_{i-1}\}} \widetilde{V}_{N}(\widetilde{S}_{f}^{j,s}) - \widetilde{V}_{N}(\widetilde{S}_{f}^{*})\bigg],\nonumber
   		\end{equation}
   		subject to $ R_{k}^{d_i} = R_F/N_S $. The optimal solution can be derived similar to Lemma \ref{Lemma:Segment}. 
   	\item  Step 2: Let $
   		d_{i^*}=\arg \min\limits_{d_i} Q_{d_i,s}^* $, the solution of Problem \ref{prob:proactive} is given by $ P_{f,s}^k =  P_{f,s}^{k,d_{i^*}}$, $ N_{f,s}^k = N_{f,s}^{k,d_{i^*}} $, and $  \mathcal{A}_f^k = \{ d_i|\forall i \geq i^*\} $.
   	 \end{itemize}
\end{Algorithm}
	
\section{Simulation}\label{sec:sim}

\begin{figure} 
	\centering 
	\subfigure[$ N_C = 20 $]{ \label{fig:20caches}  \includegraphics[width=3in]{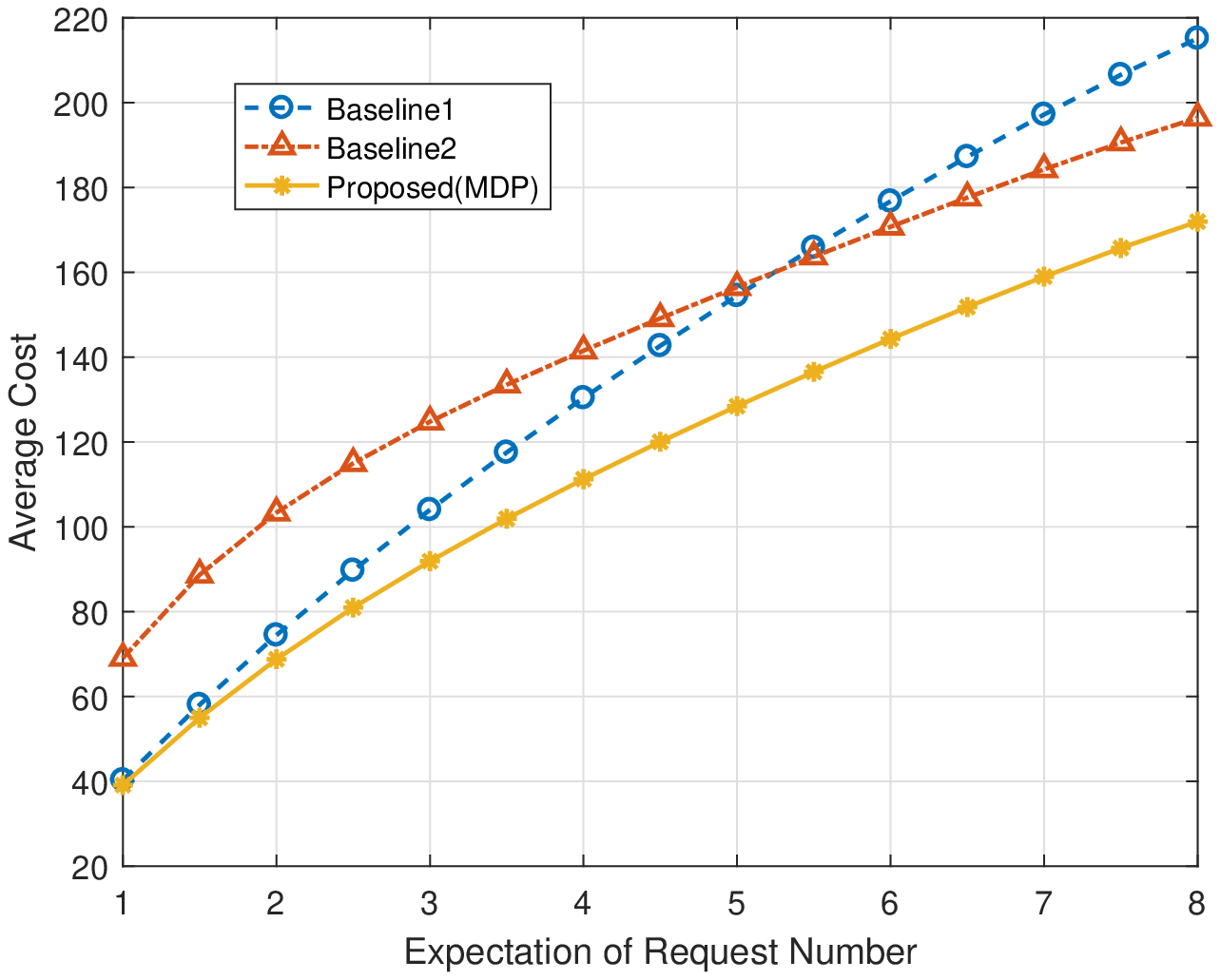}} 
	\hspace{0.2in} 
	\subfigure[$ N_C = 25 $]{ \label{fig:25caches} 
	\includegraphics[width=3in]{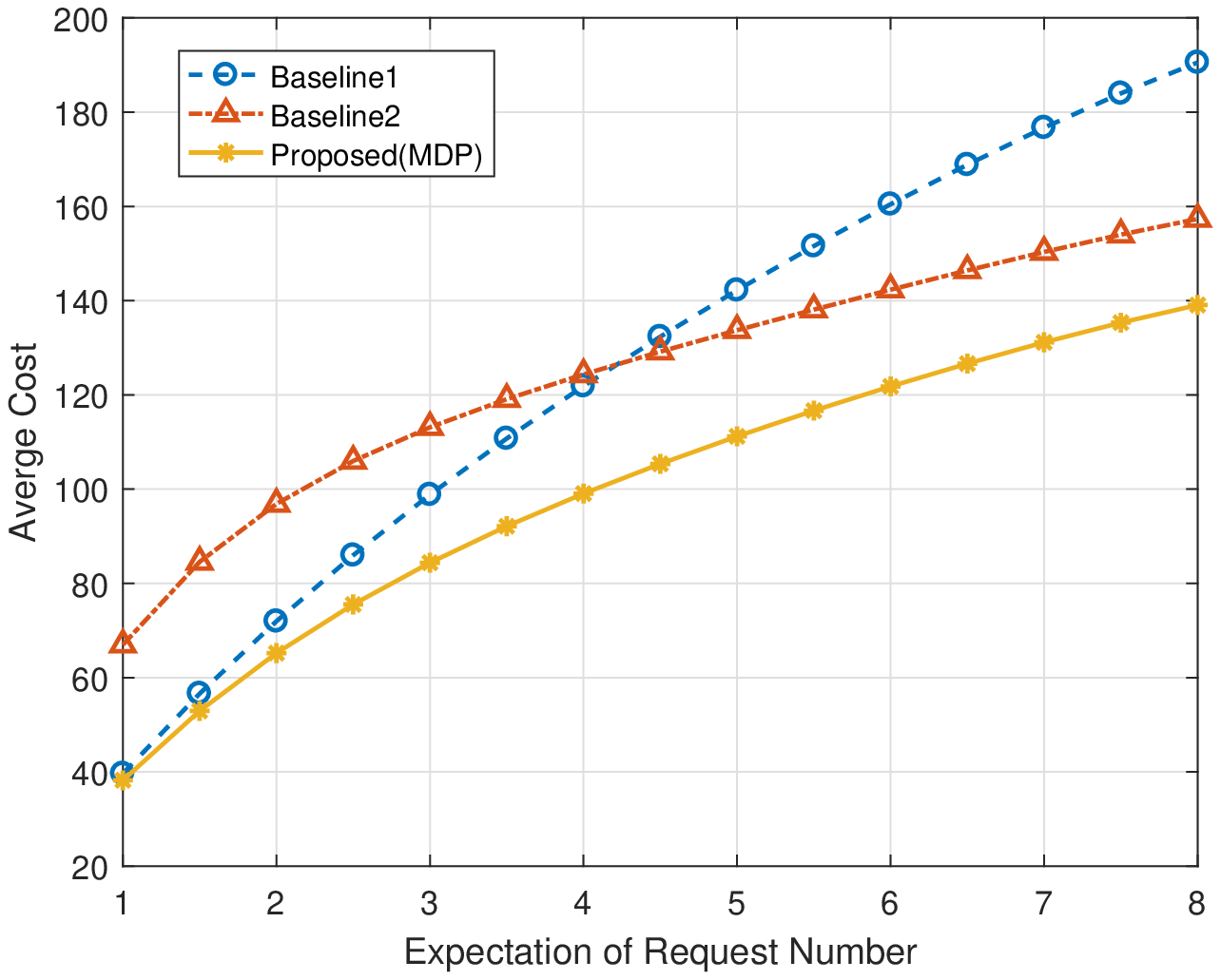}} 
	\caption{The average total cost versus the average times of request in file lifetime, where the number of cache nodes is $ 20 $ and $  25 $ respectively.} 
	\label{fig:caches} 
\end{figure}

In the simulation, the cell radiusis $ 500 $ meters, cache nodes are randomly deployed on the cell-edge region with a service radius of $ 90 $ meters. The number of antennas at the BS is 8. The downlink path loss exponent is $3.5$. The standard derivation of shadow fading in 6dB. The file size is $ 140 $Mb, and is further divided into $ 10 $ segments. The transmission bandwidth is $ 20 $MHz. The weights on transmission energy and time are $w_e=1$ and $w_t=100$ respectively. The performance of the proposed algorithm will be compared with the following two baselines.
\begin{Baseline}
	The BS only ensures the segment delivery to the requesting users in each transmission. The cache nodes with better channel condition to the BS can also decode the file segments.
\end{Baseline}
\begin{Baseline}The BS ensures that all the cache nodes can decode the downlink file in the first transmission. Hence, all the cache nodes can help to forward the file since the second file request.
\end{Baseline}

The performance of the proposed low-complexity algorithm (Algorithm \ref{alg:AMDP}) is compared with the above two baselines in Fig.\ref{fig:caches}. In the simulation, the number of cache nodes is $ 20 $ and $25$ respectively, and the requesting users are uniformly distributed in the cell coverage, and distribution statistics is known to the BS. Hence, the analytical expressions derived in Section \ref{sub:app} can be used to calculate the approximated value functions. It can be observed that the proposed Algorithm \ref{alg:AMDP} is superior to the two baselines for any expected number of requests per file lifetime. Moreover, the Baseline 1 has better performance than Baseline 2 when the popularity of the file is high (larger expected number of file requests). 

\begin{figure} 
	\centering 
	\subfigure[$\widetilde{S} = \{\mathcal{B}^i_{f,s} = 0, \forall f,s, i\} $]{ \label{fig:bound_full}  
		\includegraphics[width=3in]{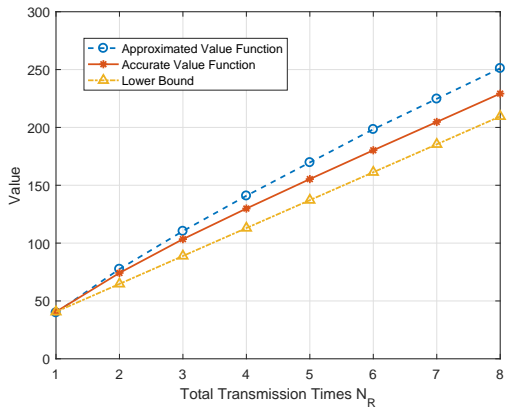}} 
	\hspace{0.2in} 
	\subfigure[$\widetilde{S} = \{\mathcal{B}^i_{f,s} = 0, \forall f,s, i=1,2,...,10\} \cup \{\mathcal{B}^i_{f,s} = 1, \forall f,s, i=11,12,...,20\}$]{ \label{fig:bound_half}
		\includegraphics[width=3in]{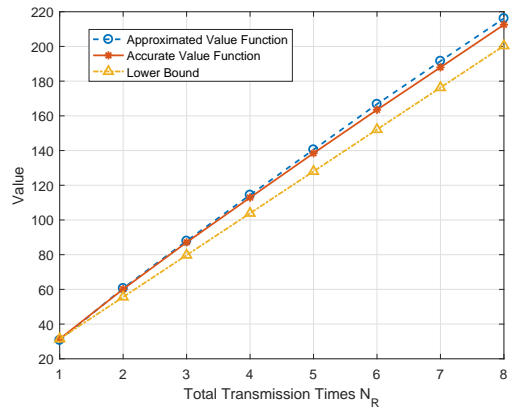}} 
	\caption{Illustration of value function and its bounds, where $ N_C = 20 $.} 
	\label{fig:bound} 
\end{figure}

The approximation error of the value function versus different stages is illustrated in Fig.\ref{fig:bound}, where the true value function and the bounds derived in Lemma \ref{lem:bound} are plotted. Fig. \ref{fig:bound_full} shows the state where all the cache nodes are empty, and Fig. \ref{fig:bound_half} shows the state where half of cache nodes have decoded the whole file. It is shown that for both states, both upper and lower bounds are tight, and therefore approximation error is small. 

\begin{figure} 
	\centering 
	\subfigure[$ 3 $ hot zones]{ \label{fig:learning(3)}  
		\includegraphics[width=3in]{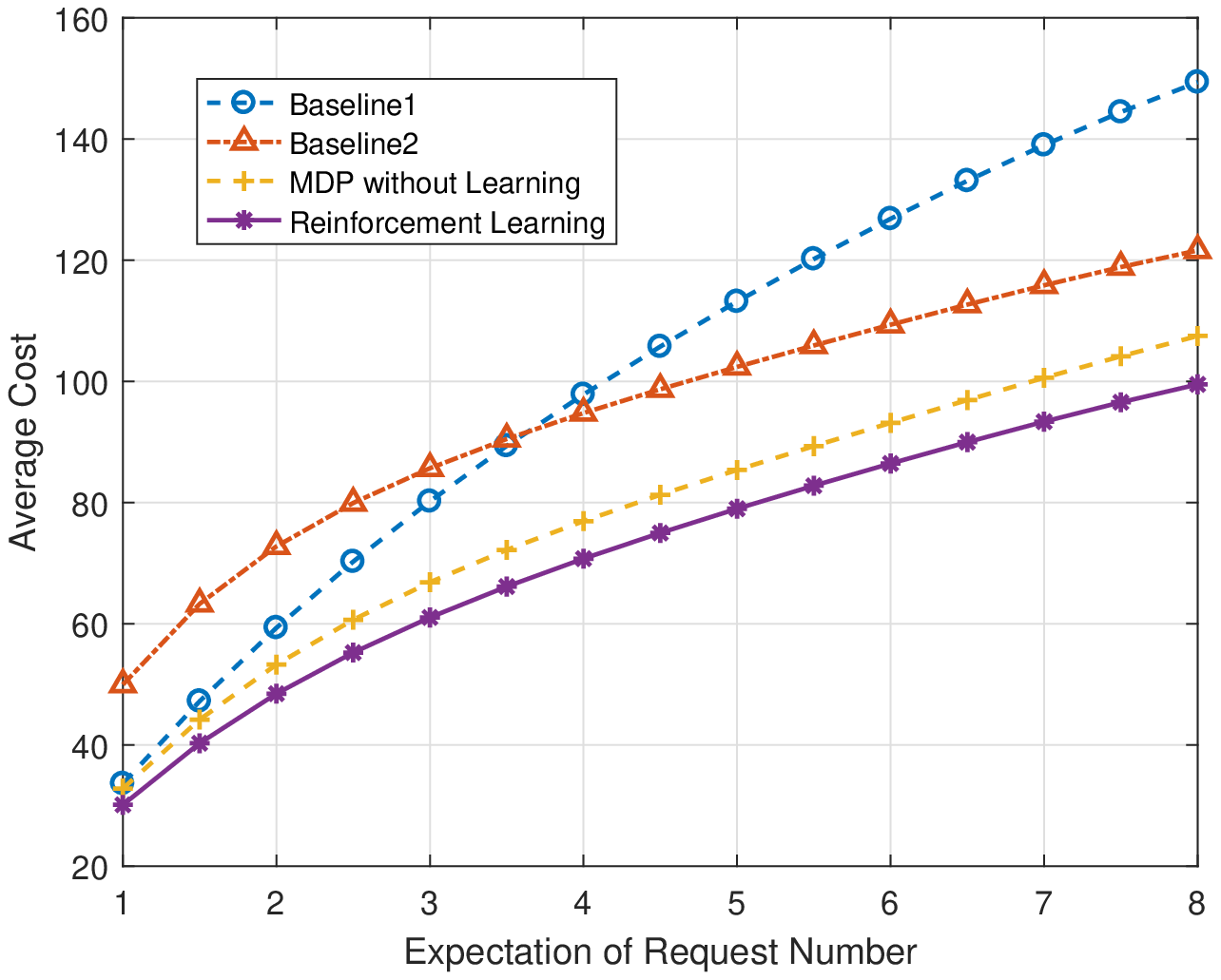}}
	\hspace{0.2in} 
	\subfigure[$ 4 $ hot zones]{ \label{fig:learning(4)}
		\includegraphics[width=3in]{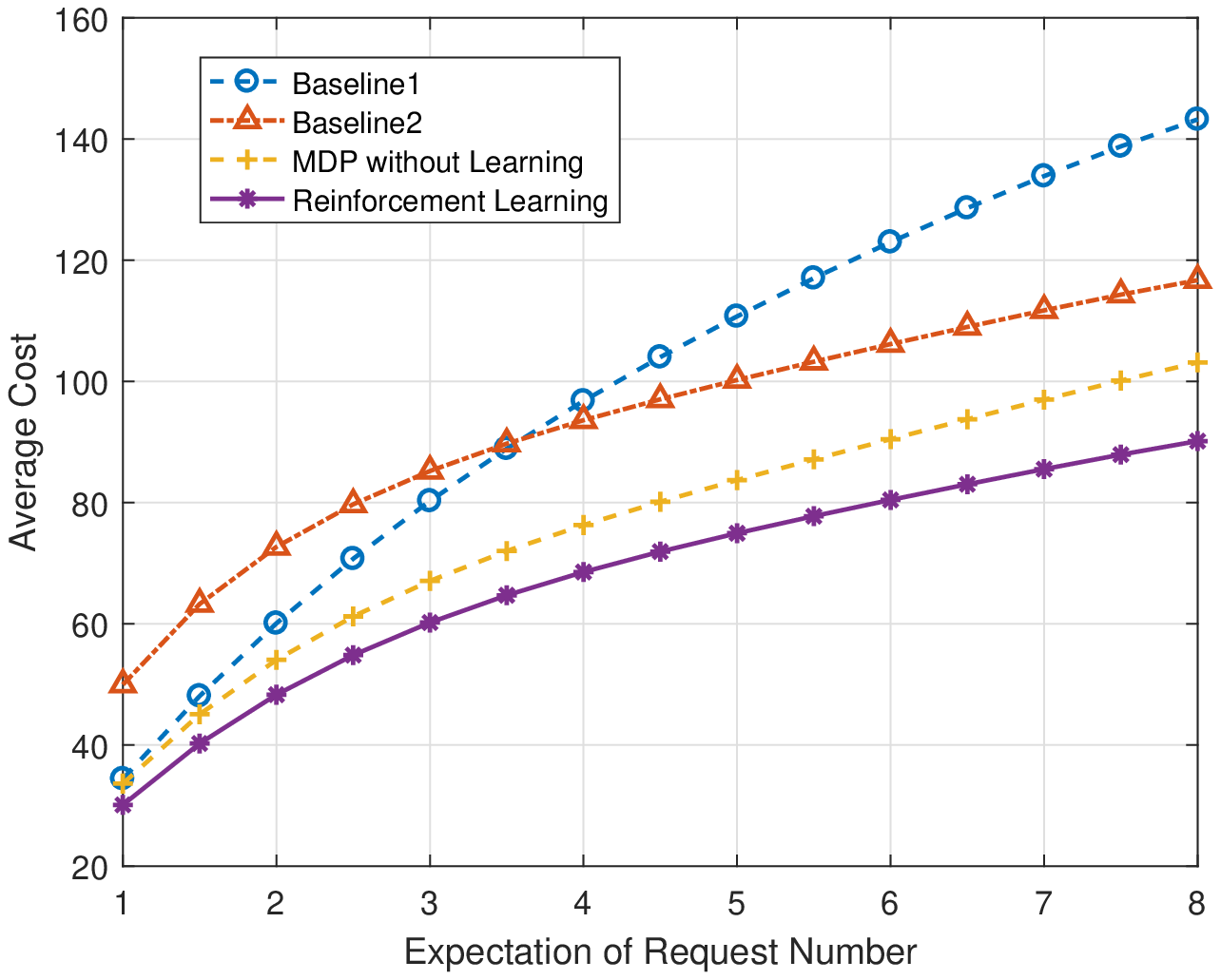}} 
	\caption{The average total cost versus the expectation of request times, where exits 3 hot zones in the cell.} 
	\label{fig:learn} 
\end{figure}

In Fig. \ref{fig:learn}, there are $ 3 $ and $ 4 $ hot zones in the cell coverage, each with radius $ 90 $ m. The statistics of users' distribution is unknown to the BS. The probability that the user appears on the one hot zone is $12.5\%$ (larger than the other regions). The performance of two baselines, the proposed Algorithm \ref{alg:AMDP} assuming users are uniformly distributed, and the proposed Algorithm \ref{alg:AMDP} with learning-based evaluation of value functions (Algorithm \ref{alg:learning}) are compared. It can be observed that the proposed  learning algorithm has the best performance. Moreover, the performance gain of the learning-based algorithm is larger for more hot zones.

\begin{figure}[tb]
	\centering
	\includegraphics[scale = 0.6]{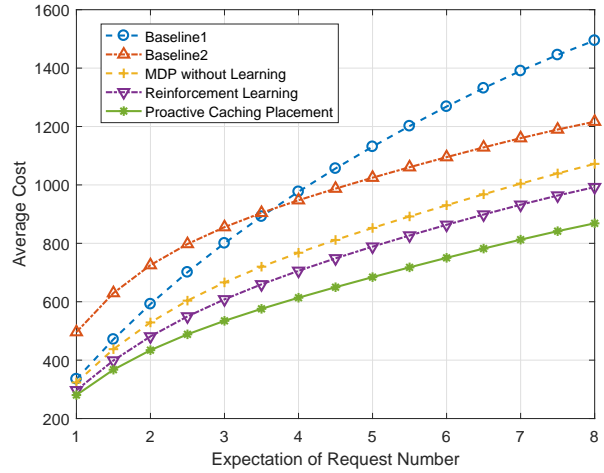}
	\caption{The Performance of Proactive Caching Placement}
	\label{fig:proactive}
\end{figure}

Finally, the performance of proactive caching placement is demonstrated in Fig. \ref{fig:proactive}, where there are 10 files, 3 hot zones in the cell, 50000 times of proactive transmission opportunities in the file's lifetime. The performance of the proactive caching placement algorithm is compared with the above two baselines, Algorithm \ref{alg:AMDP} assuming uniform user distribution, and Algorithm \ref{alg:AMDP} with learning-based evaluation of value functions. It can be observed that the proposed proactive caching placement algorithm can further improve the offloading performance, compared with the algorithms without it.

\section{Conclusion}\label{sec:con}
We consider the downlink file transmission with the assistance of cache nodes in this paper. The number of requests of one file within its lifetime is modeled as a Poisson point process, and the downlink resource minimization problem without proactive caching placement can be formulated as a MDP with random number of stages. We first propose a revised Bellman's equation, where the optimal control policy can be derived. In order to address the curse of dimensionality, we also introduce a low-complexity sub-optimal solution based on linear approximation of value functions. The approximated value function can be calculated analytically with the knowledge of distribution statistics of users. In practice, the statistics of the distribution may be unknown to the BS. We continue to propose a learning-based online algorithm to evaluate the value functions. Furthermore, we derive a bound on the gap between the approximated value functions and the real value functions. Finally, we propose a proactive caching placement algorithm, which can exploit the channel temporal diversity of shadowing effect. It is shown by numerical simulation that the proposed algorithms can significantly reduce the resource consumption at the BS, and the proactive caching placement algorithm can further improve the performance.

\section*{Appendix A: Proof Of Lemma \ref{Lemma:reduce_space}}
Let $\widetilde{V}_{N_R-n}(\widetilde{S}_{f,n})=\mathbb{E}_{\eta,\rho}[V_{N_R-n}(S_{f,n+1})] $, where the expectation is taken over the randomness of shadowing and requesting users' pathloss, we have
\begin{eqnarray}
&&\sum\limits_{S_{f,n+1}}V_{N_R-n}(S_{f,n+1})Pr(S_{f,n+1}|S_{f,n},\Omega_{f,n})=\sum\limits_{\widetilde{S}_{f,n+1}}{\widetilde{V}_{N_R-n}(\widetilde{S}_{f,n+1})Pr(\widetilde{S}_{f,n+1}|{S}_{f,n},\Omega_f)}\nonumber
\end{eqnarray}
Taking expectation with respect to the shadowing and pathloss on (\ref{eqn:bellman-fix}), we have
\begin{align}
\widetilde{V}_{N_R-n+1}(\widetilde{S}_{f,n}) &= \mathbb{E}_{\eta,\rho} \bigg\{ \min_{\Omega(S_{f,n})} \sum_{s} g_{f,n,s}(S_{f,n},\Omega_{f,n}) +\! \sum\limits_{S_{f,n+1}}{V_{N_R-n}(S_{f,n+1})Pr(S_{f,n+1}|S_{f,n},\Omega_{f,n})} \bigg\}\nonumber \\
&=\min_{\Omega(\widetilde{S}_{f,n})} \mathbb{E}_{\eta,\rho} \bigg\{\sum_{s} g_{f,n,s}(S_{f,n},\Omega_{f,n})+\sum\limits_{\widetilde{S}_{f,n+1}}{\widetilde{V}_{N_R-n}(\widetilde{S}_{f,n+1})Pr(\widetilde{S}_{f,n+1}|{S}_{f,n},\Omega_f)} \bigg\} \nonumber.
\end{align}

\section*{Appendix B: Proof Of Lemma \ref{Lemma:Segment} }
First of all, we have the following high SNR (signal-to-noise ratio) approximation on the throughput $ R_{f,n,s} $.
\begin{eqnarray}
R_{f,n,s} &\approx& N_{f,n,s} \mathbb{E}_{\mathbf{h}_{f,n,s}} \left[ \log_2 \left(  \frac{||\mathbf{h}_{f,n,s}||^2 P_{f,n,s}}{N_T \sigma^2_z} \right) \right]=N_{f,n,s}[\theta+\log_2(P_{f,n,s})]\nonumber. 
\end{eqnarray}
With $R_{f,n,s}= \frac{R_F}{N_S}$, we  have $
N_{f,n,s}=\frac{R_F}{N_S[\theta+\log_2(P_{f,n,s})]}$.
Hence the original optimization becomes
$$
\min \limits_{P_{f,n,s} }  \frac{R_F(w_eP_{f,n,s}+w_t)}{N_S[\theta+\log_2(P_{f,n,s})]} \nonumber
$$  
Taking first-order derivative on $f(P_{f,n,s})$, the optimal transmission power $ P_{f,n,s}^* $ can be obtained.

\section*{Appendix D: Proof Of Lemma \ref{Lemma:learning} }
We only prove the convergence of $ \widetilde{V}^t_{N_R-n+1}(\widetilde{S}_{f}^*) $, and the convergence of $ \widetilde{V}^t_{N_R-n+1}(\widetilde{S}_{f}^{i,s}) $ can be applied similarly.
Let $$\varepsilon_t = \widetilde{V}_{N_R-n+1}(\widetilde{S}_{f}^*) - (N_R-n+1)I(\mathbf{l}_{g,m} \notin \mathcal{C})
\sum_s (w_e P_{g,m,s}^*N_{g,m,s}^* + w_t N_{g,m,s}^*)$$ denotes the estimate error in $t$-th iteration. It is clear that the estimation errors are i.i.d. with respect to $ t $,  $\mathbb{E} [\varepsilon_t] =0$ and $Var [\varepsilon_t ] < +\infty$. Note that  $\widetilde{V}^t_{N_R-n+1}(\widetilde{S}_{f}^*)$ can be written as
\begin{eqnarray}
&&\widetilde{V}^t_{N_R-n+1}(\widetilde{S}_{f}^*)=\sum_{i=0}^{t} \frac{\widetilde{V}_{N_R-n+1}(\widetilde{S}_{f}^*)-\varepsilon_i}{t+1}=\widetilde{V}_{N_R-n+1}(\widetilde{S}_{f}^*)-\sum_{i=0}^{t}\frac{\varepsilon_i}{t+1} ,\nonumber
\end{eqnarray}
where the total estimate error is $\sum_{i=0}^{t}\frac{\varepsilon_i}{t+1}$. The mean and variance of total estimate error are summarized below:
\begin{eqnarray} 
&& \mathbb{E}\bigg\{     \sum_{i=0}^{t}\frac{\varepsilon_i}{t+1}                 \bigg\}=0, \nonumber \\
 && Var\bigg\{     \sum_{i=0}^{t}\frac{\varepsilon_i}{t+1}\bigg\} = \frac{Var[\varepsilon_i ]}{t+1}.    \nonumber     
\end{eqnarray} 
When $t\rightarrow +\infty $,  the variance of estimation error tends to zero, and $ \widetilde{V}^t_{N_R-n+1}(\widetilde{S}_{f}^*) $ converges to $ \widetilde{V}_{N_R-n+1}(\widetilde{S}_{f}^*) $.

\section*{Appendix E: Proof Of Lemma \ref{lem:bound} }
\subsubsection{Proof of Upper Bound}
The approach of mathematical induction will be used in the proof. Thus without loss of generality, we shall assume that  the upper bound holds when the first $ l $-th cache nodes have not decoded the $ (f,s) $-th segment, and prove that the upper bound also holds when the first $ (l+1) $-th cache nodes have not decoded the $ (f,s) $-th segment.
Define the system state $  \widetilde{T_f}^{c,s} = [\mathcal{B}^i_{f,j}=1,\forall j\neq s, \forall i] \cup[\mathcal{B}^i_{f,s}=0,\forall i=1,2,\cdots,c]\cup[\mathcal{B}^i_{f,s}=1,\forall i > c]$

\begin{itemize}
	\item Step 1: When $c=1$, the upper bound holds as follows
	\begin{eqnarray}
	\widetilde{V}_{N_R-n+1}(\widetilde{T_f}^{c,s}) = \widetilde{V}_{N_R-n+1}(\widetilde{S}_{f}^*) +\bigg( \widetilde{V}_{N_R-n+1}(\widetilde{S}_{f}^{c,s})-\widetilde{V}_{N_R-n+1}(\widetilde{S}_{f}^*)\bigg) \nonumber
	\end{eqnarray}
	
	\item Step 2: Suppose the following bound holds for $ c=l $
		\begin{equation}\label{eqn:Assumption1}
		\widetilde{V}_{N_R-n+1}(\widetilde{T_f}^{l,s}) \leq  \widetilde{V}_{N_R-n+1}(\widetilde{S}_{f}^*)+ \sum_{j=1,2,\cdots,l} \bigg( \widetilde{V}_{N_R-n+1}(\widetilde{S}_{f}^{c,j})-\widetilde{V}_{N_R-n+1}(\widetilde{S}_{f}^*)\bigg)
		\end{equation}
	
	\item Step 3: When $c=l+1$, we can apply the following sub-optimal control policy: (1) if the requesting users appear in the coverage of $ \mathcal{C}_1 \cup \mathcal{C}_2 \cup ... \cup \mathcal{C}_l $, the optimal control policy for system state $ \widetilde{T_f}^{l,s} $ is applied; (2) if the requesting users appear in the coverage of $ \mathcal{C}_{l+1} $, the optimal control policy for system state $ \widetilde{S}_f^{l+1, s} $ is applied; (3) if the requesting users appear outside the coverage of any cache nodes, choose the one from the above two policies with larger transmission resource consumption. Let $\breve{V}_{N_R-n+1}(\widetilde{T_f}^{l+1,s})  $ be the average cost of the above sub-optimal control policy, we have
	\begin{equation*}
	\widetilde{V}_{N_R-n+1}(\widetilde{T_f}^{l+1,s}) \!\leq\! \breve{V}_{N_R-n+1}(\widetilde{T_f}^{l+1,s}) \!\leq\! \widetilde{V}_{N_R-n+1}(\widetilde{T_f}^{l,s}) + \bigg( \widetilde{V}_{N_R-n+1}(\widetilde{S}_{f}^{c,l+1})-\widetilde{V}_{N_R-n+1}(\widetilde{S}_{f}^*)\bigg).
	\end{equation*}

\end{itemize}

Although the above proof is for the $ (f,s) $-th file segment, it can be trivially extended to arbitrary file segments. Thus the upper bound is proved.

\subsubsection{Proof of Lower Bound}
Let $\Omega_{f,n}^*$ be the optimal control policy, we have 
\begin{eqnarray}
&&\widetilde{V}_{N_R-n+1}(\widetilde{S}_{f,n})-\widetilde{V}_{N_R-n+1}(\widetilde{S}_{f}^*)\nonumber \\
&&=\mathbb{E}_{\eta,\rho} \bigg\{ \sum_{s} g_{f,n,s}(\widetilde{S}_{f,n},\Omega_{f,n}^*)\bigg\} +\mathbb{E}_{\eta,\rho} \bigg\{\sum\limits_{\widetilde{S}_{f,n+1}}{\widetilde{V}_{N_R-n}(\widetilde{S}_{f,n+1})Pr(\widetilde{S}_{f,n+1}|{S}_{f,n},\Omega_{f,n}^*)}\bigg\} \nonumber \\
&&-\mathbb{E}_{\eta,\rho} \bigg\{ \sum_{s} g_{f,n,s}(\widetilde{S}_{f}^*,\Omega_{f,n}^*)\bigg\}-\widetilde{V}_{N_R-n}(\widetilde{S}_{f}^*) \nonumber.
\end{eqnarray}
Because $
\mathbb{E}_{\eta,\rho} \bigg\{ \sum_{s} g_{f,n,s}(\widetilde{S}_{f,n},\Omega_{f,n}^*)\bigg\}\geq \mathbb{E}_{\eta,\rho} \bigg\{ \sum_{s} g_{f,N_R,s}(\widetilde{S}_{f,n},\Omega_{f,N_R}^*)\bigg\}$\\  and $\mathbb{E}_{\eta,\rho} \bigg\{ \sum_{s} g_{f,n,s}(\widetilde{S}_{f}^*,\Omega_{f,n}^*)\bigg\}= \mathbb{E}_{\eta,\rho} \bigg\{ \sum_{s} g_{f,N_R,s}(\widetilde{S}_{f}^*,\Omega_{f,N_R}^*)\bigg\}, $
We have
\begin{eqnarray}
&&\mathbb{E}_{\eta,\rho} \bigg\{ \sum_{s} g_{f,n,s}(\widetilde{S}_{f,n},\Omega_{f,n}^*)\bigg\} -\mathbb{E}_{\eta,\rho} \bigg\{ \sum_{s} g_{f,n,s}(\widetilde{S}_{f}^*,\Omega_{f}^*)\bigg\}\nonumber \\
&&\geq\mathbb{E}_{\eta,\rho} \bigg\{ \sum_{s} g_{f,N_R,s}(\widetilde{S}_{f,n},\Omega_{f,N_R}^*)\bigg\} -\mathbb{E}_{\eta,\rho} \bigg\{ \sum_{s} g_{f,N_R,s}(\widetilde{S}_{f}^*,\Omega_{f,N_R}^*)\bigg\}\nonumber \\
&&=\sum_{{\{(i,s)|\forall \mathbf{B}^i_{f,s}(\widetilde{S}_{f,n}) = 0\}}} \bigg( \widetilde{V}_{1}(\widetilde{S}_{f}^{i,s})-\widetilde{V}_{1}(\widetilde{S}_{f}^*)\bigg).\nonumber
\end{eqnarray}
We also have 
\begin{eqnarray}
\mathbb{E}_{\eta,\rho} \bigg\{\sum\limits_{\widetilde{S}_{f,n+1}}{\widetilde{V}_{N_R-n}(\widetilde{S}_{f,n+1})Pr(\widetilde{S}_{f,n+1}|{S}_{f,n},\Omega_{f,n}^*)}\bigg\} -\widetilde{V}_{N_R-n}(\widetilde{S}_{f}^*)\geq 0. \nonumber
\end{eqnarray}
As a result, the lower bound is straightforward.

\bibliographystyle{IEEEtran}
\bibliography{icc}

\end{document}